\documentclass[11pt,hidelinks]{article}

\usepackage{fullpage}
\title{Subgame Optimal and Prior-Independent Online Algorithms}
\author{Jason Hartline \thanks{Department of Computer Science, Northwestern University. Email:hartline@northwestern.edu.} \and Aleck Johnsen \thanks{Geminus Research. Email:aleck@geminusresearch.com.} \and Anant Shah \thanks{Department of Computer Science, Northwestern University. Email:anantshah2026@u.northwestern.edu.}}
\date{}

\RequirePackage{natbib}
\usepackage{amsthm,amsmath,amssymb}
\usepackage{ifthen}
\usepackage{graphicx}
\usepackage{hyperref}
\usepackage{cleveref}
\usepackage{subcaption}
\usepackage{diagbox}
\usepackage{comment}

\usepackage[dvipsnames]{xcolor}
\usepackage{soul}

\usepackage{bbm}
\def\1{\mathbbm{1}}

\DeclareMathOperator*{\argmax}{arg\,max}
\DeclareMathOperator*{\argmin}{arg\,min}

\newtheorem{theorem}{Theorem}[section]
\newtheorem{definition}{Definition}

\newtheorem{lemma}[theorem]{Lemma}
\newtheorem{fact}[theorem]{Fact}

\newtheorem{corollary}[theorem]{Corollary}

\newtheorem{property}{Property}

\usepackage[ruled,vlined]{algorithm2e}
\SetKw{KwBy}{by}

\usepackage{algorithmic}

\newcommand{\twopartdef}[4]
{
	\left\{
		\begin{array}{ll}
			#1 & \mbox{if } #2 \\
			#3 & \mbox{if } #4
		\end{array}
	\right.
}

\newcommand{\threepartdef}[6]
{
	\left\{
		\begin{array}{lll}
			#1 & \quad & \mbox{if } #2 \\
			#3 & \quad & \mbox{if } #4 \\
                #5 & \quad & \mbox{if } #6
		\end{array}
	\right.
}

\usepackage{tikz}

\makeatletter
\newcommand*{\encircled}[1]{\relax\ifmmode\mathpalette\@encircled@math{#1}\else\@encircled{#1}\fi}
\newcommand*{\@encircled@math}[2]{\@encircled{$\m@th#1#2$}}
\newcommand*{\@encircled}[1]{%
  \tikz[baseline,anchor=base]{\node[draw,circle,outer sep=0pt,inner sep=.2ex] {#1};}}
\makeatother

\newcommand{\floor}[1] {\left \lfloor #1 \right \rfloor}
\newcommand{\ceil}[1]{\left \lceil #1 \right \rceil}
\renewcommand{\qedhere}{}

\newcommand{\pidistclass}{\mathcal{F}}
\newcommand{\pialgclass}{\mathcal{A}}
\newcommand{\skidistclass}{\mathcal{P}}
\newcommand{\piinputclass}{\mathcal{X}}

\newcommand{\pidistclassepsgrid}{\pidistclass_{\epsilon}}
\newcommand{\skidistclassepsgrid}{\skidistclass_{\epsilon}}

\newcommand{\OPT}{\text{OPT}}
\newcommand{\pidist}{F}
\newcommand{\pidisteps}{\pidist_{\epsilon}}
\newcommand{\skidist}{p}
\newcommand{\pidistoverdist}{\pidist'}
\newcommand{\pidistoverdisteps}{\pidist_{\epsilon}'}
\newcommand{\pidistclassrem}[1]{\mathcal{F}_{#1}}
\newcommand{\pidistepsindex}{k}
\newcommand{\pidistoverdistvar}{F^{*}}
\newcommand{\pialg}{A}
\newcommand{\pirandalg}{\pialg'}
\newcommand{\pirandalgepseq}{\pialg^{*}_{\epsilon}}
\newcommand{\pidistoverdistepseq}{\pidist^{*}_{\epsilon}}
\newcommand{\pialgbestresponse}{\pialg_{\pidistoverdist}}
\newcommand{\pirandalgvar}{\pialg^{*}}
\newcommand{\eqrandalg}{\pialg''}

\newcommand{\bayopt}{\OPT_{\pidist}}
\newcommand{\bayopteps}{\OPT_{\pidisteps}}
\newcommand{\piopt}{\OPT_{\pidistclass}}
\newcommand{\skibayopt}{\OPT_{\skidist}}

\newcommand{\skibayoptdistoverdistcomp}[1]{(\OPT_{\skidistoverdist})_{#1}}
\newcommand{\skibayoptspec}[1]{\OPT_{\skidist_{#1}}}

\newcommand{\piinput}{\textbf{X}}
\newcommand{\piinputmod}{\textbf{X}'}
\newcommand{\piinputsubseq}[1]{\textbf{X}^{#1}}
\newcommand{\inplen}{T}
\newcommand{\inpelement}[1]{X_{#1}}
\newcommand{\inpelementmod}[1]{X_{#1}'}
\newcommand{\timeindex}{i}
\newcommand{\subtimeindex}{j}
\newcommand{\piratio}{\beta}
\newcommand{\skidisteps}{\skidist_{\epsilon}}
\newcommand{\skigrid}[1]{G_{#1}}

\newcommand{\seqkrmod}{\mathbf{X'}^{k}_{r}}
\newcommand{\util}{u}
\newcommand{\maxutil}{\Bar{u}}
\newcommand{\epscoverscaling}{\zeta}

\newcommand{\numactions}{m}
\newcommand{\timestep}{t}
\newcommand{\timehorizon}{T}
\newcommand{\subtimestep}{t'}
\newcommand{\regret}{R}
\newcommand{\probdist}{P}
\newcommand{\avgprobdist}{\Bar{\pidistoverdist}}
\newcommand{\avgalgdist}{\Bar{\pialg}}

\DeclareMathOperator{\support}{supp}
\newcommand{\supp}[1]{\support(#1)}

\newcommand{\alg}[1]{\pialg_{#1}}

\newcommand{\stopcost}{B}
\newcommand{\skidistoverdist}{P}
\newcommand{\stp}{S}
\newcommand{\continue}{C}

\newcommand{\pdf}{f}
\newcommand{\contcost}[2]{c^{\skidist}_{#1}(#2)} 
\newcommand{\contcostp}[3]{c^{\skidist}(#1,#2,#3)} 

\newcommand{\suppsize}[1]{|\supp{#1}|}
\newcommand{\decompfunc}{g_{k}}
\newcommand{\pimodalg}{A'}

\DeclareMathOperator{\poly}{poly}

\begin{document}

\maketitle

\begin{abstract}

This paper takes a game theoretic approach to the design and analysis of online algorithms and illustrates the approach on the finite-horizon ski-rental problem. This approach allows beyond worst-case analysis of online algorithms. First, we define “subgame optimality” which is stronger than worst case optimality in that it requires the algorithm to take advantage of an adversary not playing a worst case input. Algorithms only focusing on the worst case can be far from subgame optimal. Second, we consider prior-independent design and analysis of online algorithms, where rather than choosing a worst case input, the adversary chooses a worst case independent and identical distribution over inputs. Prior-independent online algorithms are generally analytically intractable; instead we give a fully polynomial time approximation scheme to compute them. Highlighting the potential improvement from these paradigms for the finite-horizon ski-rental problem, we empirically compare worst-case, subgame optimal, and prior-independent algorithms in the prior-independent framework. 
\end{abstract}

\section{Introduction}
\label{s:intro}

We consider online algorithm design, an area where traditional analysis has looked at the worst case competitive ratio between the performance of the online algorithm (which processes components of the input one at a time and does not know the full input in advance) and the performance of the optimal offline algorithm (which knows the input sequence in advance). We develop two notions of beyond worst-case analysis that give better than worst-case guarantees on natural inputs \citep[cf.][]{Roughgarden20}.  On non-worst-case inputs there is slack between the worst-case bound and the best performance of an algorithm -- {\em subgame optimality} requires the algorithm to perform optimally with regard to this slack.  This is a stronger notion than the usual worst case notion. Another way to model inputs is to assume that their components are independent and identically distributed -- {\em prior-independent optimality} looks at worst-case algorithms over such distributions. This is weaker than the usual worst-case notion.  The comparison of these three approaches (including worst-case) gives more perspective on online algorithms beyond the worst case.

We define {\em subgame optimality} for online algorithms as a stronger requirement than that of traditional competitive analysis. Like competitive analysis, the quantity of interest is the ratio of the algorithm's cost to the optimal offline cost. However, we additionally require that the algorithm optimize this ratio in the subgame induced by any history of its past decisions and adversary inputs. In our setting, the adversary is \textit{non-adaptive}. \footnote{A non-adaptive adversary must choose the input sequence in advance without any knowledge of the online decisions taken by the algorithm.} For any subgame that is not worst-case, the algorithm achieves a better ratio. Every subgame optimal algorithm is also worst-case optimal, but the opposite is not generally true and in fact there can be a very large gap between a subgame optimal algorithm and (a non-subgame optimal) worst-case optimal algorithm for a fixed non-worst-case input (see \Cref{s:sub-gameoptalgs}). 

We define {\em prior-independent optimality} for online algorithms as a competitive analysis with stronger informational assumptions. Prior-independent analysis is a robust framework that gives guarantees in worst case over a class of distributions of inputs. In our analysis, as is common in the literature, the distributions are such that distinct coordinates of the input are generated independently and identically. The objective is to minimize the ratio of costs between an online algorithm that does not know from which distribution in the class the input is drawn to the optimal online algorithm that does know the input distribution. \citet{HR-08} initiated the prior-independent setting within mechanism design. Recent papers have adapted the framework to other domains like online
learning (discussed with related work).  

Analysis of optimal prior-independent algorithms, generally, has not led to closed form characterizations. The way forward considered by this paper is to identify conditions on the online algorithms problem which admit the existence of efficient algorithms for computing a near-optimal prior-independent algorithm. Properties of robust algorithms could be identified using these computational methods. Moreover, it focuses attention on understanding which prior-independent algorithm design problems are computationally tractable.

Our analysis views a robust online algorithm design problem generally as a two-player zero-sum game between an {\em algorithm designer} and an {\em adversary}. This perspective enables the application of the following known tools from game theory in the context of online algorithms:
\begin{itemize}
\item Online optimization is a sequential game and the sequential equilibrium concept of {\em subgame optimality} requires the algorithm to play perfectly even on non-worst-case inputs.\footnote{In the literature ``subgame perfect equilibria'' are ones where all players play subgame optimal strategies.}

\item Equilibria in two-player zero-sum games are efficiently computable via online learning or the ellipsoid method even when one player has a very large action space (\citealp{HLP-19}).

\item All equilibria of a two-player zero-sum game have the same value (payoff to one player), and this value is guaranteed for even sub-optimal strategies of the opposing player. In mixed equilibria, the payoff of any action taken with positive probability is the same. Equilibria under restricted action spaces can be lifted to equilibria with unrestricted action spaces when general actions offer no benefit to either player.
\end{itemize}

We apply the concepts of subgame optimality and prior-independence to the finite-horizon ski-rental problem. In this problem, the algorithm observes a sequence of days, each having good or bad weather. Every day after observing the weather, it makes a decision of whether to continue or stop (a.k.a., rent or buy). If it stops, it suffers a final cost of $\stopcost$, whereas if it continues, it suffers a cost of one if the weather was good, while it suffers a cost of zero if the weather was bad. The optimal offline algorithm either pays $B$ on the first day, or always rents and pays $1$ for each good-weather day, whichever is cheaper. The
challenge to the online algorithm comes from the fact that it does not know the number of good weather days in advance.

\subsection{Results and Technical Contributions}
\label{ss:resultsintro}

\paragraph{Subgame Optimal Online Algorithms.}

First, we characterize an optimal worst-case algorithm in the finite-horizon ski-rental problem. The time horizon plays a crucial role in the decisions taken. For a large time horizon, specifically $T \gtrsim 2B$, the optimal strategy for the infinite-horizon ski-rental problem (\citealp{KMMO-94}) is optimal. The optimal worst-case algorithm is obtained by solving for equilibrium in a game that ignores the bad weather days. Solving for equilibrium in this restricted game is tractable. Next, we characterize the subgame optimal algorithm via a reduction to the worst-case algorithm problem. Finally, we compare the performance of each of these algorithms, showing that the difference in their performances on non-worst-case inputs can be as bad as the worst-case approximation ratio. Further discussion is in \Cref{s:sub-gameoptalgs}.

\paragraph{Prior-Independent Online Algorithms.}

The paper provides a general framework to compute near-optimal
prior-independent online algorithms for settings that satisfy certain properties (listed below). Observe that the class of algorithms and distributions have infinite cardinality. We reduce the problem of computing such an algorithm
to the problem of computing an approximate Nash equilibrium
in a finite zero-sum game (\citealp{Khachiyan-79} and \citealp{FS-99}). Our technique takes the point of view of the adversary. The framework computes the prior-independent optimal algorithm by learning a worst-case distribution over inputs. The
four properties that the technique requires are as
follow:

\begin{itemize}
 \item (Small-Cover) There exists an $\epsilon$-cover of the class of distributions with small size such that for every fixed algorithm, the elements of the cover approximate the adversary's objective in an $\epsilon$-ball around them. 
\item (Efficient Best Response) For any distribution over
  the cover of distributions that the adversary plays, the optimization problem of minimizing the ratio between an algorithm's cost and the
  Bayesian optimal cost can be solved efficiently.
    \item (Efficient Utility Computation) For the computed optimal algorithm; for each distribution in the cover of distributions, the ratio of its expected performance to that of the optimal algorithm which knows the distribution can be computed efficiently.
    \item (Bounded Best Response Utility) For the computed optimal algorithm; for each distribution in the cover of distributions, the ratio of the algorithms expected performance to that of the optimal algorithm which knows the distribution is bounded and is polynomial in the input size.
\end{itemize}

\noindent
Our prior-independent result states that if these properties are satisfied, then a near-optimal prior-independent algorithm can be computed
efficiently.

We apply the above computational framework to the finite-horizon ski-rental problem. In this problem, each day is good weather with probability $p \in [0,1]$, independent over days. We show that the four properties of our framework are satisfied. Further discussion is in \Cref{s:framework}.

In follow up work, \cite{guo2024algorithmic} utilize the framework we propose in this paper to efficiently compute a near-optimal prior-independent information aggregator for a binary state with conditionally independent signals. The main technical challenge solved in their work is to identify a small-cover of the class of information structures.

\paragraph{Comparisons Across Models.} 

We empirically compare the algorithms in the prior-independent framework. \Cref{fig:comp_approximationratio_intro} plots the approximation ratio of each algorithm against the worst-case distribution over inputs in the prior-independent framework. In the prior-independent framework, the subgame optimal algorithm outperforms the worst-case optimal algorithm. The comparison of the three algorithms is against the worst-case distribution over distributions. Another comparison would have been across their performance in the prior-independent framework. Such a comparison does not differentiate the worst-case and sub-game optimal algorithm. We provide sufficient conditions on the prior-independent environment for when this is the case. Further discussion is in \Cref{s:simulations}. 

\begin{figure}
    \centering
    \includegraphics[scale=0.45]{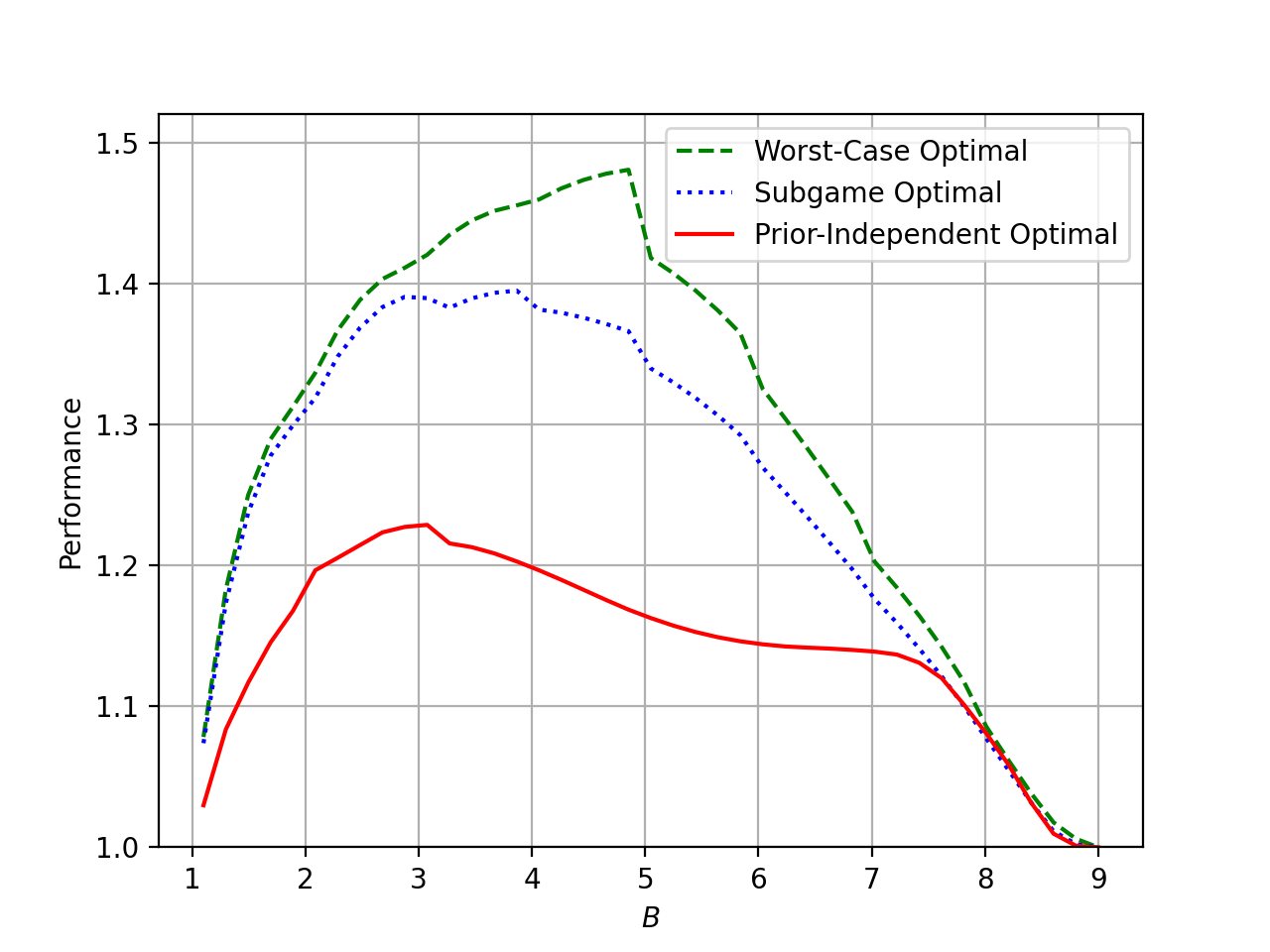}
    \caption{The time horizon is fixed at $T=9$. The \textit{stopping} cost is varied continuously from $[1,T]$. The figure plots the approximation ratio of each algorithm against the worst-case distribution of the adversary in the prior-independent framework.}
    \label{fig:comp_approximationratio_intro}
\end{figure}

\subsection{Related Work}
\label{ss:relatedwork}

Two recent works give prior-independent results that apply within online algorithm design.  \citet{HJ-21} give a general method for proving lower bounds on prior-independent algorithms, albeit their strongest results apply for the 2-input case, which corresponds to just 2 rounds in an online learning problem.
\citet{HJL-20} connect prior-independent design generally to a worst-case {\em benchmark design problem} and further, as part of their application to the problem of no-regret expert learning, prove that the follow-the-leader algorithm is prior-independent-optimal.

The prominent theme of beyond worst-case analysis considers the problem of online algorithms with
distributional information. \citet{KP-94} consider a diffuse adversary
model in which the performance of an algorithm is the worst-case
ratio, over distributions, of the expected performance of the
algorithm to the expected performance of the optimal offline
algorithm. In contrast to their setting, the benchmark in the prior-independent framework is an optimal online algorithm with full knowledge of the distribution over inputs. \citet{DEHLS-17} consider the $k$-server problem with
distributional advice. The requests at each time instant arrive from
some known sequence of independent distributions. The objective they consider is to come up with an algorithm whose
expected performance approximates the performance of the optimal
algorithm who knows the distribution. While they compete against the Bayesian optimal algorithm, their setting is not prior-independent as the approximation algorithm depends on the distribution of requests. \citet{MNS-12} consider a stochastic
setting for online algorithms. They provide a framework for general
resource allocation problems which incorporate guarantees from two
algorithms:  one that works well for specified distributional assumptions and the
other being a classical worst-case online algorithm, giving the designer a
best of both worlds guarantee.

The ski-rental problem has been looked at from a distributional viewpoint as well. \citet{MP-11} consider a two-stage game for the ski-rental problem with distributions. \citet{DKTVZ-21} consider a setting similar to ours, in which they assume that the good weather days are distributed according to a log-concave distribution. Crucially, however, the algorithm they propose depends on samples from the distributions. The problem we consider can be thought of as a zero-sample setting of \citet{DKTVZ-21}. 

Our analysis of prior-independent online algorithms can be viewed as giving a reduction from computing a near-optimal
prior-independent algorithm to computing an approximate Nash
equilibrium in an infinite-dimensional two-player zero-sum games.
Following the result of \citet{Khachiyan-79},
finding an equilibrium in zero-sum games is tractable in the size of
the payoff matrix. There is a lot of work in the literature which looks at
computing an approximate Nash equilibrium using online-learning
techniques. If both the players are employing no-regret learning
strategies (e.g., \citealp{FS-97}), then their average play converges
to a Nash equilibrium at rate $O(\frac{1}{\sqrt{T}})$. \citet{DDK-11}
give a no-regret learning algorithm which when used by both players
simultaneously, their average play converges to an approximate Nash
equilibrium at rate $O(\frac{\log T}{T})$ and in the adversarial case
achieves the $O(\frac{1}{\sqrt{T}})$ rate. 

In general, the space of distributions and the space of algorithms could potentially be infinite. Thus, the prior-independent setting is generally an infinite action two-player zero-sum game. Concurrent work of \citet{AADDF-23} computes an equilibrium in a two-player zero-sum game where each player has infinite actions with general payoffs under the assumption that each player has access to a best-response oracle. \citet{AADDF-23} motivate computing an equilibrium in infinite games from applications in machine learning. Their bounds, however, are exponential in the desired approximation; thus not directly applicable to the setting of online algorithms. 

This manuscript is presented as follows. We give a quick preliminary on robust algorithm design in \Cref{s:prelims}. Subgame optimal algorithms are defined and analyzed in \Cref{s:sub-gameoptalgs}. A general framework to compute a near-optimal prior-independent online algorithm is provided in \Cref{s:framework}. The algorithms studied in this paper are empirically compared in the prior-independent framework, shown in \Cref{s:simulations}. All proofs are deferred to the appendix. 
\section{Robust Algorithm Design}
\label{s:prelims}

Let the number of inputs to the online algorithm design be $T$. The space of inputs is $\piinputclass^{T}$ and a particular input is denoted by $\piinput$. The input $\piinput$ is a tuple of $\inplen$ elements $(X_{1},\dots,X_{\inplen})$, where each element $X_{i} \in \mathcal{X}$. Let the class of distributions, over the space of inputs, be denoted by $\pidistclass$, where $\pidistclass \subset \Delta(\piinputclass^{\inplen})$. A distribution in this class is denoted by $\pidist$. The class of feasible deterministic algorithms is denoted as $\pialgclass$ and a particular algorithm in this class as $\pialg$. The expected performance of an algorithm $\pialg$ for a particular distribution $\pidist$ is denoted as $\pialg(\pidist) = \mathbf{E}_{\piinput \sim \pidist}[\pialg(\piinput)]$. A randomized algorithm $\pirandalg \in \Delta(\pialgclass)$ is a convex combination over deterministic algorithms and its expected performance is $\pirandalg(\pidist) = \mathbf{E}_{\pialg \sim \pirandalg}[\pialg(\pidist)]$. We now describe the Bayesian optimal analysis framework and the robust algorithm design framework.


\begin{definition}
The Bayesian optimal algorithm design problem is given by a distribution $\pidist$ and a family of algorithms $\pialgclass$ and asks for the algorithm $\bayopt$ with the minimum expected cost
$$\text{OPT}_{F}(F) = \argmin_{\pialg \in \pialgclass} \pialg(\pidist).$$
\end{definition}
It is without loss to restrict the Bayesian optimal algorithm design problem to deterministic algorithms. For any randomized algorithm, there exists a deterministic algorithm in its support that obtains a weakly smaller cost. We now define the robust algorithm design framework.  

\begin{definition}
The robust algorithm design problem is given by a family of distributions $\pidistclass$, a family of algorithms $\pialgclass$ and solves the minimization problem
$$\piratio = \min_{\pirandalg \in \Delta(\pialgclass)} \max_{\pidist \in \pidistclass}\frac{A'(F)}{\text{OPT}_{F}(F)}.$$
\noindent
The algorithm which solves it is denoted as $\piopt$. 
\end{definition}

Solving this optimization problem would guarantee a multiplicative approximation factor $\piratio$ against the Bayesian optimal performance for any distribution from the feasible class of distributions. 

The robust algorithm design problem can be viewed as a two-player zero-sum game between an algorithm player, whose strategy set is $\pialgclass$, and an adversary, whose strategy set is $\pidistclass$. 
The robust algorithm optimization problem solves for a mixed Minmax strategy profile in this game. Playing the robust optimal algorithm guarantees that no matter what distribution the adversary plays, the agent will be a $\beta$ multiplicative factor close to the expected performance of the optimal algorithm which knows the distribution.

Our focus will be on two special cases of the robust algorithm design problem. First, the worst-case algorithm design problem and second, the prior-independent algorithm design problem. Each of them having different distributional assumption over the space of inputs. 

\paragraph{Worst-Case Algorithm Design.} The class of distributions are degenerate distributions over the space of inputs, i.e., $$\mathcal{F} := \{\text{Pr}[\mathbf{X}]=1 : \mathbf{X} \in \mathcal{X}^{\inplen}\}.$$

\paragraph{Prior-Independent Algorithm Design.} A distribution in $\mathcal{F}$ is such that, for $D \in \Delta(\mathcal{X})$, $X_{i} \sim D$, i.i.d, across all $i \in [T]$, i.e., $$\mathcal{F} := \{X_{i} \sim D, \text{ i.i.d } \forall i \in [T] : D \in \Delta(\mathcal{X})\}.$$

We utilize equilibrium properties in game theory to arrive at solutions to the robust algorithm design problem. These concepts are fairly standard and provided in \Cref{s:appendixprelims} for completeness. 
\section{Subgame Optimal Algorithms}
\label{s:sub-gameoptalgs}

In this section, we focus on the worst-case algorithm design paradigm. We introduce the notion of {\em subgame optimal} online algorithms. Subgame optimal algorithms, while optimal in worst-case, take advantage if the adversary makes a mistake. In \Cref{ss:stratandextensive}, we formally introduce the notion of a subgame optimal algorithm, showing how it differs from an optimal worst-case algorithm. In \Cref{ss:finitehorizonskirental}, we apply these notions to the finite-horizon ski-rental problem.

\subsection{Strategic Form and Extensive Form Representation}
\label{ss:stratandextensive}

The worst-case algorithm design paradigm solves for the value $$\piratio = \min_{\pirandalg \in \Delta(\pialgclass)} \max_{\mathbf{X}}\frac{A'(\mathbf{X})}{\text{OPT}_{\mathbf{X}}(\mathbf{X})},$$ where $\text{OPT}_{\mathbf{X}}(\mathbf{X})$ is the cost of the optimal algorithm which knows the complete sequence $\mathbf{X}$.

Solving for an optimal worst-case algorithm reduces to solving for a mixed Nash equilibrium in a zero-sum game. The game is between an algorithm player and an adversary. The strategy space of the algorithm player is space of all deterministic algorithms, while the strategy space of the adversary is the space of all input sequences. Our focus will be on a non-adaptive adversary. A non-adaptive adversary chooses the input sequence in advance, without knowledge of the decisions taken by the algorithm. A worst-case algorithm can be found by solving for equilibrium in the strategic form representation of the game. The framework assumes that the {\em adversary} gives a worst-case distribution over inputs to this algorithm. Such an input is said to be \textit{on-path}. Any other input that performs strictly worse than the worst-case is said to be \textit{off-path}. A consequence of the assumptions of this framework is that the worst-case algorithm need not perform optimally off-path. A desirable worst-case algorithm should be able to take advantage of mistakes made by the {\em adversary}. This leads to the question of considering stronger equilibrium concepts.

A natural representation of an online problem is through its extensive form representation (see Chapter 6 in the textbook by \citealp{BE-98} for an introduction to the different representations of a game). An algorithm obtained through the strategic form game need not take optimal decisions at each decision node in the extensive form game. A \textit{subgame optimal} strategy in the extensive form representation captures the notion of optimality off-path. Such a strategy takes an optimal decision at every node in the decision tree.

\begin{definition}[Subgame Optimal Algorithm]
    Consider an online algorithm design problem and its extensive form representation. An algorithm is said to be subgame optimal if it is a subgame optimal strategy in this game. Such an algorithm is denoted by $\text{SG-OPT}$. 
\end{definition} With this notion, we can define a new benchmark for the algorithm design problem. It may be computationally intractable to identify SG-OPT. The benchmark allows for approximations to the performance of the sub-game optimal algorithm. \begin{definition} 
[Subgame Optimal Benchmark]
\label{d:sgoptbenchmark}
    The performance of an algorithm $A$, against the subgame optimal benchmark is $$\max_{\mathbf{X}}\frac{A(\mathbf{X})}{\text{SG-OPT}(\mathbf{X})}.$$
\end{definition}

We apply these concepts to the finite-horizon ski-rental problem in \Cref{ss:finitehorizonskirental}.






\subsection{Finite-Horizon Ski-Rental Problem}
\label{ss:finitehorizonskirental}

In this section we describe the finite horizon ski-rental problem. An agent observes an online sequence of days, each having either good or bad weather. The agent has knowledge about the length of the sequence, denoted by $\inplen$, but observes the weather of each day in an online fashion. We represent the sequence as $\piinput \in \{0,1\}^{\inplen}$ where the $i^{th}$ component of $\piinput$, represented by $\inpelement{i}$, is the weather on the $i^{th}$ day. For any $i \in [\inplen]$, $\inpelement{i}=1$ corresponds to the $i^{th}$ day having good weather while $\inpelement{i}=0$ corresponds to the $i^{th}$ day having bad weather. Upon observing the weather, the agent needs to decide whether to \textit{stop} for a fixed, known cost $\stopcost \in \mathbb{R}^{+}$, or \textit{continue}, suffering a normalized cost of $1$ on a good weather day while suffering a cost of $0$ on a bad weather day. This decision needs to be made without knowledge of the future sequence of days. Once the agent decides to \textit{stop}, they do not incur any further costs. The goal of the agent is to come up with an algorithm be competitive against the optimal hindsight cost, i.e., the strategy with the minimum cost had the sequence of days been known in advance. 

A deterministic algorithm is a tuple of functions, which map the observed sequence of days to a decision of \textit{stop} or \textit{continue}. Formally, \begin{definition}
\label{d:bralgskirent}
A deterministic algorithm is a tuple of functions $(\alg{1},\alg{2},\dots,\alg{\inplen})$ where

$$\alg{i} : \left(\inpelement{1},\dots,\inpelement{\timeindex}\right) \to \{\stp,\continue\} \quad \forall \left(\inpelement{1},\dots,\inpelement{\timeindex}\right) \in \{0,1\}^{\timeindex} \quad \forall \timeindex \in [\inplen]$$

\noindent
where $\stp$ is a decision to stop and $\continue$ is a decision to continue.
\end{definition}
A deterministic algorithm defines a deterministic action taken by the agent for each prefix of histories. For a fixed $\inplen$, the number of deterministic algorithms that an agent can play is finite. We denote the set of deterministic algorithms by $\mathcal{A}$. A randomized algorithm is a randomization over deterministic algorithms. The cost that a particular deterministic algorithm $\pialg$ suffers along a sequence of days $\mathbf{X}$ is denoted by $A(\mathbf{X})$. It is the sum of the continuation costs and the cost of stopping, if the algorithm were to stop on that particular sequence. 

%
The agent wants to find the randomized algorithm which minimizes the competitive ratio. Formally, it wants to solve the optimization problem $$\min_{\pirandalg \in \Delta(\mathcal{A})} \max_{\piinput \in \{0,1\}^{\inplen}}\frac{\mathbf{E}_{\pialg \sim \pirandalg}[\pialg(\piinput)]}{\OPT_{\piinput}(\piinput)},$$ where $\OPT_{\piinput}(\piinput)$ is the optimal hindsight cost for the sequence of days $\piinput$. If there are $k = \sum_{i=1}^{\inplen} X_{i}$ good weather days in a sequence, then $\OPT_{\piinput}(\piinput) = \min\{k,B\}$. 

\subsubsection{Optimal Randomized Algorithm}
\label{ss:optrand}

In this section, we solve for an optimal randomized algorithm, consequently solving for the competitive ratio of the finite horizon ski-rental problem. As described before, such an algorithm can be obtained by solving for equilibrium in the strategic form representation of the zero-sum game between algorithm player and adversary. We denote this game by \textit{SRP}. A challenge that arises upon viewing it through this lens is that the strategy spaces of the adversary and the algorithm player are exponential in $\inplen$. In equilibrium, the adversary player could potentially mix good weather days with bad weather days and the algorithm player would need to account for such input sequences. 

A key result, which we prove, is that it is without loss for the adversary to provide a contiguous sequence of good weather days followed by a sequence of bad weather days.

To that end, we define a \textit{reduced} two-player zero-sum game, denoted as \textit{SRP-R}. In \textit{SRP-R}, the pure strategy space of the adversary can be indexed by the contiguous number of good days followed by bad days. Formally, let the pure strategy space of the adversary be denoted by $\mathcal{F}_{g}$. Define this class of strategies as $$\mathcal{F}_{g} := \{F^{k} : k \in \{1,2,\dots,T\}\},$$ where $F^{k}$ denotes the input which gives a contiguous sequence of $k$ good weather days followed by bad weather days. The pure strategy space of the algorithm player is indexed by the number of good weather days for which the player will \textit{continue} after which they will stop. Denoting the set of pure strategies of the algorithm player as $\mathcal{A}_{g}$, we define this set as $$\mathcal{A}_{g} := \{A^{l} : l \in \{0,1,\dots,\inplen\}\},$$ where $A^{l}$ corresponds to the algorithm which \textit{continues} for the first $l$ good weather days after which it \textit{stops} on the $(l+1)^{\text{st}}$ good weather day. A deterministic algorithm in $\mathcal{A}_{g}$ implicitly defines the decision it takes on all possible input sequences; the algorithm \textit{ignores} any bad weather day by \textit{continuing} and only keeps a track of the number of good weather days. The pure strategy $A_{T}$ will \textit{continue} for all days.

Let $u\left(A^{l},F^{k}\right)$ be the utility to the adversary player for algorithm strategy $A^{l}$ and adversary strategy $F^{k}$. We have that $$u(A^{l},F^{k}) := \twopartdef{\frac{l+B}{\min\{k,B\}}}{k > l}{\frac{k}{\min\{k,B\}}}{k \leq l}.$$ Note that algorithm $A^{\inplen}$ always \textit{continues}; while the input sequence $F^{\inplen}$ corresponds to all good weather days. 

To show that it is without loss for the adversary to consider input sequences from $\mathcal{F}_{g}$, it suffices to prove that any mixed strategy Nash equilibrium in \textit{SRP-R} is a mixed strategy Nash equilibrium in \textit{SRP}. It is straightforward that the algorithm player would not deviate from it's strategy given the strategy of the adversary. The \textit{off-path} decisions of deterministic algorithms in $\mathcal{A}_{g}$ lead to a scenario in which the adversary player does not deviate from it's strategy in the equilibrium of the reduced game.

\begin{lemma}
    \label{l:truerandworstcase}
   Any mixed Nash equilibrium in \textit{SRP-R} is a mixed Nash equilibrium in \textit{SRP}.
\end{lemma}

To identify the optimal randomized algorithm, it thus suffices to solve for the mixed Nash equilibrium strategy of the algorithm player in \textit{SRP-R}.

First, observe that the set of pure strategies $\{F^{k} : k \in \{\ceil{B},\dots,\inplen - 1\}\}$, is weakly dominated by $F_{\inplen}$. For any strategy $F^{k}$ where $k \in \{\ceil{B},\dots,\inplen\}\}$, the optimal hindsight cost is $B$. Among these strategies, $F^{\inplen}$ makes any strategy of the algorithm suffer a weakly larger cost with more good days. As far as the algorithm player is concerned, the set of pure strategies $\{A^{l} : l \in \{\inplen - \floor{B}, \dots, \inplen - 1\}\}$ are weakly dominated by $A^{\inplen}$. For any input sequence in $\mathcal{F}_{g}$, the algorithm player will never want to \textit{stop} at any of the last $\floor{B}$ time indices. It will suffer a smaller cost by \textit{continuing}. 

Thus, to find an equilibrium in \textit{SRP-R}, it is without loss to consider the zero-sum game where the space of pure strategies for the algorithm is $$\mathcal{A}'_{g} = \{A^{l} : l \in \{0,\dots,T-\floor{B}-1\} \cup \{\inplen\}\},$$ and the space of pure strategies for the adversary player is $$\mathcal{F}'_{g} = \{F^{k} : k \in \{1,\dots,\ceil{B}-1\} \cup \{\inplen\}\}.$$ 

\autoref{t:randworstcase} finds the strategy of the algorithm player in a mixed Nash equilibrium of \textit{SRP-R}. The proof of \autoref{t:randworstcase} is provided in \autoref{a:sub-game-proofs}. A sketch of the computation is as follows. A property of game theory is that in any mixed equilibria, the expected payoff of any pure strategy that has positive probability is the same (this payoff may differ across players). Depending on the value of $\inplen$ relative to $B$, we eliminate a set of strategies that are dominated in both $\mathcal{A}_{g}'$ and $\mathcal{F}_{g}'$. The result follows by solving the system of equations obtained by equating the expected payoffs for each player in the resulting game. 

By \Cref{l:truerandworstcase}, it follows that such a strategy is the optimal worst-case randomized algorithm for the finite horizon ski-rental problem.

\begin{lemma}
\label{t:randworstcase}
    A mixed strategy for the algorithm player in \textit{SRP-R}:
    \begin{itemize}
        \item If $\ceil{\stopcost} + \floor{\stopcost} \geq \inplen + 1$, is,  \begin{align*}
            \text{Pr}^{T}[A^{l}] &= \left(\frac{\stopcost}{\stopcost-1}\right)^{l} \times Pr[A^{0}] \quad \forall l \in \{0,1,\dots,\inplen - \floor{\stopcost}-1\}, \\
            \text{Pr}^{T}[A^{\inplen}] &= \left(\frac{\stopcost}{\inplen - \stopcost}\right) \left(\frac{\stopcost}{\stopcost-1}\right)^{\inplen - \floor{\stopcost} -1}\left(\stopcost + \floor{\stopcost}-T\right) \times Pr[A^{0}].
        \end{align*}

        \item If $\ceil{\stopcost} + \floor{\stopcost} < \inplen + 1$, is the same as the infinite-horizon ski-rental problem. 
    \end{itemize}
\end{lemma}

We now identify the competitive ratio of the finite horizon ski-rental problem. The competitive ratio corresponds to the value of \textit{SRP}. The value of the game is identified by solving for the expected utility of the adversary player, when the algorithm player uses the mixed strategy in \autoref{t:randworstcase}. The adversary is indifferent between any pure strategy in its support at equilibrium. The expected utility to the adversary player at equilibrium is thus equal to the expected utility it obtains when the adversary player gives one good weather day. 

If $\ceil{B}+\floor{B} \geq \inplen + 1$, the competitive ratio of the finite horizon ski-rental problem is $$B\text{Pr}[A^{0}] + \sum_{l=1}^{\inplen-\floor{B}-1}\text{Pr}[A^{l}] + \text{Pr}[A^{\inplen}] = \frac{1}{1-\frac{(\inplen-B)(B-1)}{B\floor{B}\left(\frac{B}{B-1}\right)^{\inplen-\floor{B}-1}}}.$$ If $\ceil{B}+\floor{B} < \inplen + 1$, the competitive ratio of the finite horizon ski-rental problem is $$B\text{Pr}[A^{0}] + \sum_{l=1}^{\ceil{B}-1}\text{Pr}[A^{l}] = \frac{1}{1-\frac{(\ceil{B}-1)}{B}\left(\frac{B-1}{B}\right)^{\ceil{B}-1}}.$$ An immediate observation is that for a fixed $B$, the competitive ratio is weakly increasing in $\inplen$. 

\begin{corollary}
\label{c:montonevalue}
    For a fixed $B$, the competitive ratio, and thus the value of \textit{SPR}, is weakly increasing in $\inplen$. 
\end{corollary}

The randomized algorithm obtained in \autoref{t:randworstcase} is an optimal worst-case algorithm. However, it performs sub-optimally if the adversary does not provide the worst-case input sequence. In the next section, we solve for a subgame optimal algorithm. Such an algorithm takes an optimal decision given an observed sub-sequence of the input. We reduce the problem of characterizing the subgame optimal algorithm to that of the optimal worst-case algorithm.

\subsubsection{Subgame Optimal Algorithm}

Consider an instance of the finite-horizon ski-rental problem with time horizon $\inplen$. It will be helpful to write the worst-case algorithm as a probability of \textit{stopping}, conditional on observing and \textit{continuing} on a sequence of contiguous good days.  On observing $k$ contiguous good weather days, let $\eta_{k}^{\inplen}$ denote the probability of \textit{stopping}, conditional on \textit{continuing} for these days. An immediate application of \autoref{t:randworstcase} gives us the following probabilities.

\begin{corollary}
\label{c:conditionprobworstcase}
    Consider an instance of the finite-horizon ski-rental problem with time horizon $\inplen$. Conditional on \textit{continuing} for the first $(k-1)$ contiguous good weather days, the probability of \textit{stopping} on the $k^{th}$ contiguous good weather day for an optimal randomized algorithm is:
\begin{itemize}
    \item If $\ceil{\stopcost} + \floor{\stopcost} \geq \inplen + 1$, for $k \in \{1,\dots,\inplen-\floor{B}\}$, $\eta_{k}^{\inplen}$, is, $$\eta^{\inplen}_{k} = \frac{1}{\frac{B\floor{B}}{\inplen-B}\left(\frac{B}{B-1}\right)^{\inplen-\floor{B}-k} -(B-1)}.$$
    \item If $\ceil{\stopcost} + \floor{\stopcost} < \inplen + 1$, for $k \in \{1,\dots,\ceil{B}-1\}$, $$\eta^{\inplen}_{k} = \frac{1}{\left(B-1\right)\left(\frac{B}{B-1}\right)^{\ceil{B}-k}\frac{B}{\ceil{B}-1}-(B-1)}.$$ Additionally, $\eta^{\inplen}_{\ceil{B}}=1$.
\end{itemize}
\end{corollary}

We now solve for a subgame optimal algorithm. To obtain a characterization for such an algorithm, an optimal choice must be made at each decision node. Let us first observe the optimal strategy of the algorithm player on observing an initial contiguous sequence of bad weather days. 

\begin{lemma}
\label{l:reducetimehorizon}
    If the adversary initially provides a contiguous sequence of $l$ bad weather days, the optimal strategy of the algorithm player implements the optimal randomized algorithm for time horizon $(\inplen-l)$. 
\end{lemma}

By \Cref{l:reducetimehorizon}, the probability of \textit{stopping} on the sequence $$\left(\underbrace{0,\dots,0,}_{l}\underbrace{1,\dots,1}_{k}\right),$$ conditional on \textit{continuing} to time index $l+k-1$, for the subgame optimal algorithm, is $\eta^{\inplen-l}_{k}$. 

It remains to characterize the subgame optimal decisions on input sequences that start with a good weather day. We show that the conditional probability of \textit{stopping} is the same as that of \textit{stopping} had all the bad weather days been provided contiguously before providing the good weather days.

\begin{lemma}
\label{l:characterizesubgame}
Consider a sub-sequence $(X_{1},\dots,X_{i},1)$ such that the total number of good weather days in $(X_{1},\dots,X_{i})$ is $k-1$. The probability of \textit{stopping} at $(X_{1},\dots,X_{i},1)$, conditional on \textit{continuing} up until this point, is $\eta_{k}^{\inplen-i+k-1}$. 
\end{lemma}

This completes the reduction from subgame optimal algorithms to worst-case optimal algorithms. We can compare the worst-case algorithm in \autoref{t:randworstcase} to the subgame optimal algorithm. The objective we use is the worst-case ratio of their performances. Let $\text{WC-OPT}$ denote a worst-case optimal algorithm and $\text{SG-OPT}$ denote a subgame optimal algorithm. The competitive ratio according to the benchmark defined in \Cref{d:sgoptbenchmark} is then $$\max_{\mathbf{X}}\frac{\text{WC-OPT}(\mathbf{X})}{\text{SG-OPT}(\mathbf{X})}.$$ 

\begin{lemma}
\label{l:subgameoptbenchmark}
    The competitive ratio against the subgame optimal benchmark for $\text{WC-OPT}$ is the value of the game, i.e., \begin{itemize}
        \item if $\ceil{B}+\floor{B} \geq T+1$, it is,  $$\frac{1}{1-\frac{(\inplen-B)}{\floor{B}}\left(\frac{B-1}{B}\right)^{\inplen-\floor{B}}}.$$
        \item if $\ceil{B}+\floor{B} < T+1$, it is, $$\frac{1}{1-\frac{(\ceil{B}-1)}{B}\left(\frac{B-1}{B}\right)^{\ceil{B}-1}}.$$
    \end{itemize}
\end{lemma}

\begin{proof}[Proof of \Cref{l:subgameoptbenchmark}]
    We can rewrite the ratio as the ratio between the competitive ratio of $\text{WC-OPT}$ and $\text{SG-OPT}$, i.e., $$\max_{\mathbf{X}}\frac{\text{WC-OPT}(\mathbf{X})}{\text{OPT}_{\mathbf{X}}(\mathbf{X})} \frac{\text{OPT}_{\mathbf{X}}(\mathbf{X})}{\text{SG-OPT}(\mathbf{X})}.$$ The proof of \autoref{t:randworstcase} tells us that the competitive ratio of $\text{WC-OPT}$ on an input sequence with $1$ good day is exactly the value of the game. This is the maximum possible value $\frac{\text{WC-OPT}(\mathbf{X})}{\text{OPT}_{\mathbf{X}}(\mathbf{X})}$ can take. If the $\mathbf{X}$ is such that the only good day is provided at the end, then $\frac{\text{SG-OPT}(\mathbf{X})}{\text{OPT}_{\mathbf{X}}(\mathbf{X})}=1$. This is the minimum possible value it can take. The lemma follows. 
\end{proof}

The result above tells us that the difference in the performance of worst-case optimal and subgame optimal can be as bad as the worst-case approximation ratio. If one expects to not see worst-case inputs, as has been argued for real-world applications \citep{Roughgarden20}, a subgame optimal algorithm gives better guarantees while maintaining the same worst-case approximation ratio.
\section{Efficient Computation Of Robust Algorithms}
\label{s:framework}

In the previous section, we were able to characterize the optimal robust algorithm. In general, obtaining such a characterization is analytically hard. Viewing it through a game theoretic lens may allow us to efficiently compute such an algorithm. Computing an optimal robust algorithm reduces to computing a mixed Nash equilibrium in the zero-sum game between algorithm and adversary. Unfortunately, the pure strategy space of the algorithm is exponentially large and the adversary has infinitely many pure strategies. Traditional techniques to compute equilibria, such as no-regret learning (\citealp{FS-99}) and the Ellipsoid algorithm (\citealp{Khachiyan-79}), cannot be applied directly to the game. We identify certain properties on the game which reduce it to a half-infinite action game with one player having access to a best response oracle. The best response oracle computes a best response for the player with infinite strategies, for any mixed strategy of the player with finite strategies, in time efficient in the support size of the mixed strategy. Finding an equilibrium in such a setting is efficient.

\begin{lemma}[\citep{FS-97},\citep{HLP-19}]
\label{l:prelimlearningexist}
    There exists an algorithm which computes an $\epsilon$-approximate Nash equilibrium in a two-player zero-sum game, where one player has finite set of actions and the other player has access to a best response oracle, in time which is polynomial in the number of actions of the first player, $1/{\epsilon}$ and the upper bound on the payoffs of the game.
\end{lemma}  

We provide sufficient conditions on the prior-independent online algorithm environment which imply the existence of efficient algorithms to compute a robust algorithm. The properties that our framework requires to hold are as follows.

\begin{property}(Small-Cover)
\label{p:frameworkutilcloseness}
Given $\epsilon > 0$, there exists a $\Theta(\epsilon$)-cover, denoted by $\pidistclassepsgrid$, of the class of distributions $\pidistclass$, such that for every $\pidisteps \in \pidistclassepsgrid$ and corresponding $\Theta(\epsilon)$-ball around $\pidisteps$ given by $\pidistclassepsgrid$,

$$\left | \frac{\pialg(\pidist)}{\bayopt(\pidist)} - \frac{\pialg(\pidisteps)}{\bayopteps(\pidisteps)} \right | \leq \frac{\epsilon}{4},$$

\noindent
for all algorithms $\pialg \in \pialgclass$. Furthermore, $\suppsize{\pidistclassepsgrid} = \poly\left(T,\frac{1}{\epsilon}\right)$, where $\inplen$ is the size of the input.
\end{property}

\begin{property}(Efficient Best Response)
\label{p:frameworkbestresponse}
There exists an algorithm which solves the optimization problem

$$\pialgbestresponse = \argmin_{\pialg \in \pialgclass}\mathbf{E}_{F \sim \pidistoverdist}\left[\frac{\pialg(\pidist)}{\bayopt(\pidist)}\right]$$

\noindent
for any $\pidistoverdist \in \Delta(\pidistclassepsgrid)$, efficiently in the size of the representation of $\pidistoverdist$ and length of the input $\inplen$.

\end{property}

\begin{property}(Efficient Utility Computation)
\label{p:frameworkutilcomputation}
For any $\pidistoverdist \in \Delta(\pidistclassepsgrid)$, for the optimal algorithm $\pialgbestresponse$ obtained by $\Cref{p:frameworkbestresponse}$, there exists an algorithm which computes, for any $\pidist \in \pidistclassepsgrid$, the utility $ \frac{\pialgbestresponse(\pidist)}{\bayopt(\pidist)}$ in run-time polynomial in the size of the input, i.e., $\text{poly}(\inplen)$. 
\end{property}

\begin{property}(Bounded Best Response Utility)
\label{p:frameworkboundedbr}
    For any $\pidistoverdist \in \Delta(\pidistclassepsgrid)$, any best response $\pialgbestresponse$ is such that, for any $\pidist \in \pidistclassepsgrid$, the utility $ \frac{\pialgbestresponse(\pidist)}{\bayopt(\pidist)}$ is at most $\text{poly}(\inplen)$.
\end{property}

Natural approaches to solve \textit{Efficient Best Response}, like backwards induction or dynamic programming, would also solve \textit{Efficient Utility Computation}. \textit{Small-Cover} intuitively says that the class of distributions $\pidistclass$ can be approximated by another class of distributions $\pidistclassepsgrid$, where $\pidistclassepsgrid \subset \pidistclass$ and the cardinality of $\pidistclassepsgrid$ is finite. Such a property would be specific to the class of distributions and algorithms considered. While it may seem as a very strong condition as it is required to hold for all $\pialg \in \pialgclass$, generally the class of algorithms can be restricted to a subset with the property that for any distribution over distributions of the adversary, an optimal best response of the algorithm lies in this class. The four properties can then be applied to that subclass. 

Further concerning \textit{Small-Cover}, it is relatively straightforward to show that for a fixed algorithm, the difference in its performance on two distributions that are close is small. However, the challenge is to show that the difference in the ratio of the performance of the algorithm to the performance of the Bayesian optimal is still small for distributions that are close. Typically, the Bayesian optimal cost can be arbitrarily small and because of this, straightforward bounds on the closeness of the ratio could be very large. Such a bound would imply a longer running time due to the finer cover required. For the application we consider in \autoref{s:model}, we show a non-trivial bound on the performance of the ratio objective, even though the Bayesian optimal cost could be arbitrarily small. With these properties, our main result is as follows.

\begin{theorem}
\label{t:apxpipoly}
If the class of algorithms $\pialgclass$ and the class of distributions $\pidistclass$ satisfy Small-Cover, Efficient Best Response, Efficient Utility Computation and Bounded Best Response Utility, an approximate robust optimal algorithm can be computed in time $\poly\left(\inplen,\frac{1}{\epsilon}\right)$ for inputs of length $\inplen$.
\end{theorem}

The main issue in computing an optimal algorithm is the size of the strategy spaces of each player. A natural direction leads to the question of computing a near-optimal robust algorithm. This problem reduces to computing an approximate Nash equilibrium in the game. The four properties aid us in solving this question by allowing us to solve for an approximate equilibrium in a reduced game. The key idea is to take the perspective of the adversary by learning a worst-case distribution over inputs. We move towards proving \autoref{t:apxpipoly} below.

\textbf{Original Game} : Pure strategy set of the algorithm is $\pialgclass$. Pure strategy set of the adversary is $\pidistclass$.

\textbf{Reduced Game} : Pure strategy set of the algorithm is $\pialgclass$. Pure strategy set of the adversary is $\pidistclassepsgrid$.

We first show that for any algorithm, for any adversary strategy in the original game, there exists a corresponding strategy in the reduced game such that their payoffs are $\epsilon$-close to one another.

\begin{lemma}
\label{l:frameworkdisctocont}
Assuming \textit{Small-Cover}, for a fixed algorithm strategy, for every adversary strategy in the original game, there exists an adversary strategy in the reduced game such that their payoffs are $\frac{\epsilon}{4}$-close to each other.
\end{lemma}

Utilizing \Cref{l:frameworkdisctocont}, we now show that an approximate Nash equilibrium in the reduced game is an approximate Nash equilibrium in the original game. 

\begin{lemma}
\label{l:frameworkepsapprox}
Assuming \textit{Small-Cover}, an $\frac{\epsilon}{4}$-approximate Nash equilibrium in the reduced game is an $\frac{\epsilon}{2}$-approximate Nash equilibrium in the original game. 
\end{lemma}

 We can then apply \Cref{l:prelimlearningexist}, on efficient computation of a $\frac{\epsilon}{4}$-approximate Nash equilibrium, to the reduced game. A proof for \autoref{t:apxpipoly} is provided below.

\begin{proof}[Proof of \autoref{t:apxpipoly}] 
Applying the algorithm in \Cref{l:prelimlearningexist}, we compute an $\frac{\epsilon}{4}$-approximate Nash equilibrium. By Efficient Best Response, Efficient Utility Computation, Bounded Best Response Utility and Small-Cover, the run-time of the algorithm in \Cref{l:prelimlearningexist} to compute an $\frac{\epsilon}{4}$-approximate Nash equilibrium will be $\text{poly}\left(\inplen,\frac{1}{\epsilon}\right)$. By \Cref{l:frameworkepsapprox}, an $\frac{\epsilon}{4}$-approximate Nash equilibrium in the reduced game will be an $\frac{\epsilon}{2}$-approximate Nash equilibrium in the original game. \autoref{f:apxnashminmax} then tells us that the mixed strategy of the algorithm player will be an $\epsilon$-approximate prior-independent optimal algorithm. 
\end{proof}

Before we look at an example in the form of the ski-rental problem, we prove a result which will be useful in restricting the class of algorithms that need to be considered while computing a robust algorithm. The result below says that it suffices to reduce to strategy space of the algorithm player to algorithms satisfying a general property.

\begin{lemma}
\label{l:reducedgameapxeqlb}
In a two-player game between player $1$, with pure strategy set $\pialgclass$, and player $2$, with pure strategy set $\pidistclass$, if for any mixed strategy that player $2$ chooses, a best response of player $1$ lies in a subset of strategies $\pialgclass^{*} \subset \pialgclass$, then an $\epsilon$-approximate Nash equilibrium in the reduced game between player $1$ and player $2$ with pure strategy sets $\pialgclass^{*}$ and $\pidistclass$ respectively is an $\epsilon$-approximate Nash equilibrium in the original game.
\end{lemma}

We apply this framework to the prior-independent finite-horizon ski-rental problem.

\subsection{Prior-Independent Finite-Horizon Ski-Rental Problem}
\label{s:model}

In this section we describe the model of the prior-independent ski-rental problem and show that it satisfies the \textit{Small-Cover}, \textit{Efficient Best Response}, \textit{Efficient Utility Computation} and \textit{Bounded Best Response Utility} property, thus implying the existence of a FPTAS to compute an approximate prior-independent optimal algorithm.

The prior-independent finite-horizon ski-rental problem makes a distributional assumption over the space of inputs in the finite-horizon ski-rental problem (described in \autoref{ss:finitehorizonskirental}). The distributional assumption over the space of inputs is as follows. Each day the weather is good with some probability $\skidist$, i.i.d over $\inplen$ days. The class of distributions $\skidistclass$ we consider is $\skidistclass = [\delta,1]$, for a $\delta > 0$, where a distribution in this class is identified with the parameter $\skidist$. We denote a distribution over probabilities as $\skidistoverdist \in \Delta(\skidistclass)$. The distribution over distributions can be thought of as sampling $\skidist \sim \skidistoverdist$ followed by instantiating the ski rental problem with this i.i.d probability of a good day.




For a fixed distribution $\skidist$, let the expected performance of the Bayesian optimal algorithm be defined as $\skibayopt(\skidist)$. For a deterministic algorithm $\pialg$, let the expected performance of the algorithm on a fixed probability $\skidist$ be defined as $A(p)$, where the expectation is over the randomness of the sequences of days generated. We can now state the prior-independent ski rental problem.

\begin{definition}
\label{d:pisrp}
The prior-independent problem is given by the class of distributions $\skidistclass = [\delta,1]$, $\delta > 0$ \footnote{We consider the range of feasible distributions to be $[\delta,1]$, for $\delta > 0$ as $\skidist=0$ is a null instance for which the Bayesian optimal performance is zero. This would lead to an undefined payoff in the prior-independent framework.}, the class of algorithms $\pialgclass$ and solves for the randomized algorithm that optimizes for the worst-case ratio $$\piratio = \min_{\pirandalg \in \Delta(\pialgclass)} \left[\max_{\skidist \in [\delta,1]} \frac{\mathbf{E}_{\pialg \sim \pirandalg}[\pialg(\skidist)]}{\skibayopt(\skidist)}\right] $$
\end{definition}

\begin{figure}[tb]
  \begin{subfigure}[t]{0.46\textwidth}
    \centering
    \includegraphics[width=\linewidth]{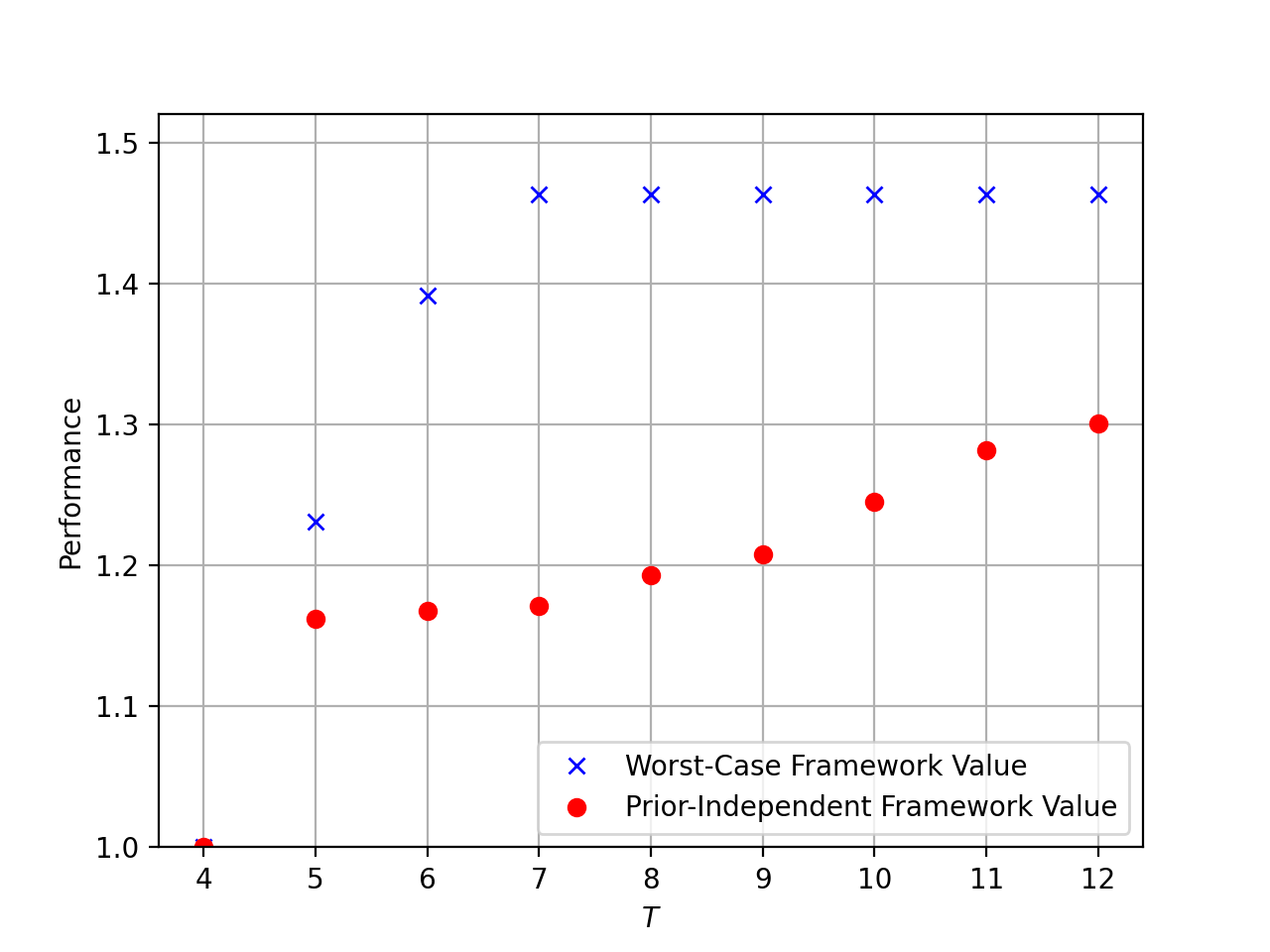} 
    \caption{Prior-Independent ratio as a function of $\inplen$} 
    \label{fig:parta}
  \end{subfigure}
  \hfill
  \begin{subfigure}[t]{0.46\textwidth}
    \centering
    \includegraphics[width=\linewidth]{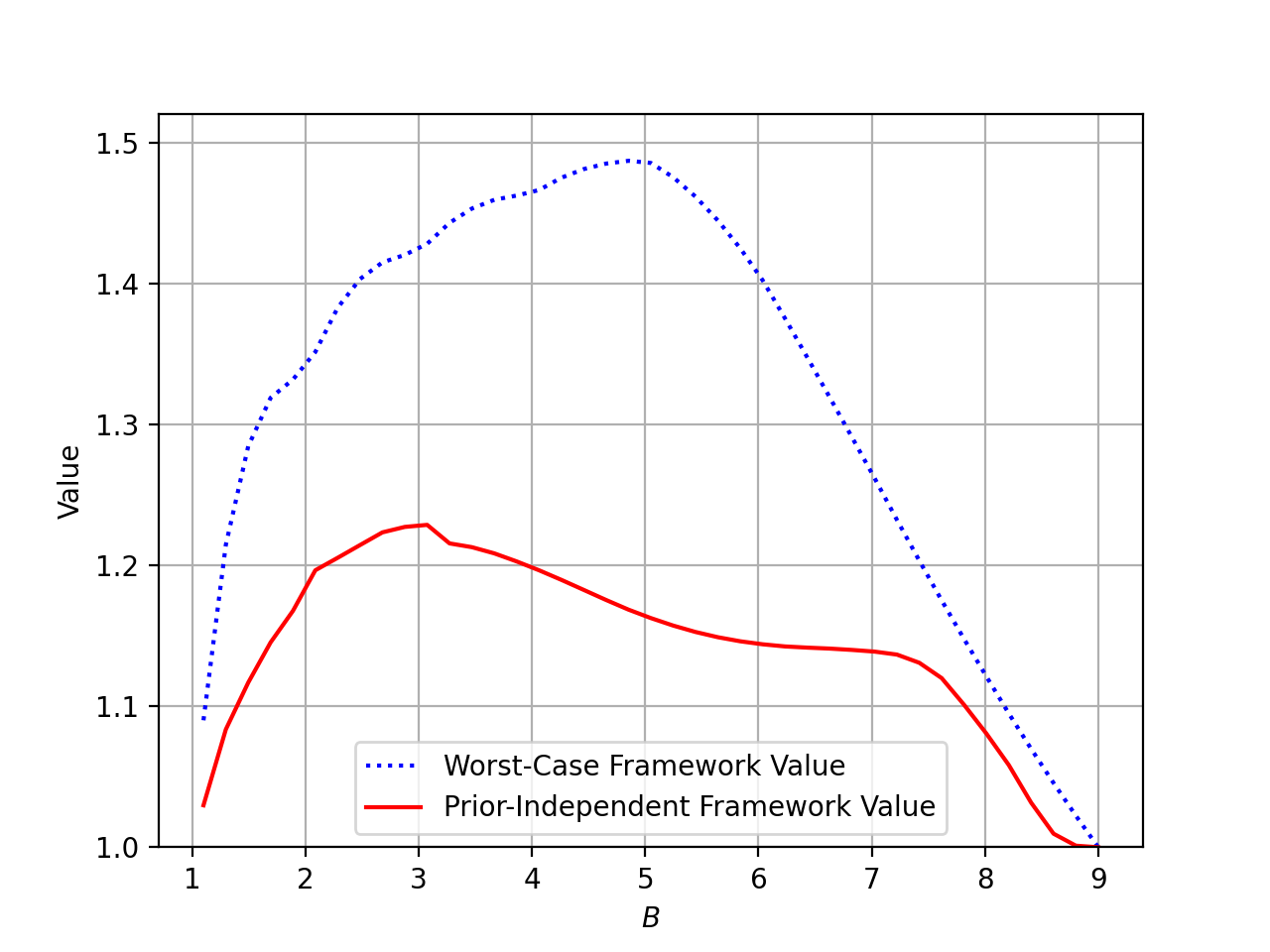} 
    \caption{Prior-Independent ratio as a function of $\stopcost$} 
    \label{fig:partb}
    \end{subfigure}
  \caption{\autoref{fig:parta} plots the prior-independent ratio as a function of integral $\inplen$ for a fixed $\stopcost=4$. \autoref{fig:partb} plots the prior-independent ratio as a function of continuous $\stopcost$ for a fixed $\inplen=9$. Compare these values to the worst-case ratio for the finite-horizon ski-rental problem as derived in \autoref{s:sub-gameoptalgs}.}
  \label{fig:piratio}
\end{figure}

\noindent
In general, the prior-independent ratio will depend on $\stopcost$ and $\inplen$. For reference, \autoref{fig:piratio} depicts the prior-independent ratio of the ski-rental problem for different instances of $\stopcost$ and $\inplen$. The approximation parameter used was $\epsilon=0.01$. We can see from the plots that the prior-independent ratios obtained are better than the worst-case ratios. The zero-sum game was approximately solved using the online learning reduction and the Multiplicative Weights Update (\citealp{FS-99}). 

We now work towards showing that the prior-independent ski-rental problem satisfies Small-Cover, Efficient Best Response, Efficient Utility Computation and Bounded Best Response Utility. \autoref{ss:bayoptchar} characterizes the Bayesian optimal algorithm and the Bayesian optimal cost for every distribution $\skidist \in [\delta,1]$. The characterization will help us in proving the Small-Cover property. We prove a property on the utilities in the zero-sum game formulation which help us prove the Small-Cover property. \autoref{ss:piskirentprop} shows that the prior-independent ski-rental problem satisfies Small-Cover, Efficient Best Response, Efficient Utility Computation and Bounded Best Response Utility, thus showing that an approximate prior-independent optimal online algorithm can be computed efficiently. 

\subsubsection{Characterization of the Bayesian Optimal Algorithm}
\label{ss:bayoptchar}

We first characterize the Bayesian optimal algorithm for a given time horizon $\inplen$, a stop cost $\stopcost > 1$, distribution $\skidist \in [\delta,1]$ and thus find $\skibayopt(\skidist)$. 

\begin{theorem}
\label{t:bayopt}
For a given time horizon $\inplen$ and stop cost $\stopcost$, the Bayesian optimal algorithm, for a distribution $\skidist \in \left[\frac{\stopcost-1}{\inplen-1},1\right]$, is to stop on a good day at time index $k$ if $p \geq \frac{\stopcost-1}{\inplen - k}$, otherwise continue, where $k \in [\min\{\ceil{\inplen-\stopcost},\inplen-1\}]$. For $\skidist \in [\delta,\frac{\stopcost-1}{\inplen-1}]$, the Bayesian optimal algorithm always continues.
\end{theorem}

The characterization in the result above helps us derive an expression for Bayesian optimal cost. For a probability $\skidist \in \left[\frac{B-1}{T-1},1\right]$, let $k$ be the last day at which the algorithm $\skibayopt$ will stop. The formula for the expected cost is 
$$\skibayopt(\skidist) = \stopcost[1-(1-\skidist)^{k}] + (1-\skidist)^{k}\skidist(\inplen-k).$$ When the probability $\skidist \in \left[\delta,\frac{\stopcost-1}{\inplen-1}\right]$, the cost is $$\skibayopt(\skidist) = \inplen \skidist.$$ It will be important to analyze the expected cost for a general algorithm $\pialg$. This will help us understand how utilities behave in the game we consider. More importantly, it will help in showing the Small-cover property for the prior-independent problem.

\begin{lemma}
\label{l:polycost}
For a given time horizon $\inplen$, a deterministic algorithm $\pialg$, and a distribution $\skidist \in [\delta,1]$, the cost $\pialg(\skidist)$ is a polynomial in $\skidist$ and has degree at most $\inplen$.
\end{lemma}

The prior-independent problem is trivial when $\stopcost \geq \inplen$. This is because no matter what the distribution $\skidist$ is, the Bayesian optimal is to always continue and thus the agent strategy of continuing for all days achieves the best possible prior-independent ratio of one. From now on we only consider $\stopcost < \inplen$.

\subsubsection{Properties of the Prior-Independent Ski-Rental Problem}
\label{ss:piskirentprop}

In this section, we show that the prior-independent ski-rental problem satisfies Small-Cover, Efficient Best Response, Efficient Utility Computation and Bounded Best Response Utility. It is without loss to restrict the class of algorithms to those which are \textit{information-symmetric} and do not stop on observing bad weather. This class satisfies the four properties, which gives us our desired result. Backwards Induction is an algorithm which for any randomization over distributions $\skidistoverdist$, solves the optimization problem
$$\min_{\pialg \in \pialgclass}\mathbf{E}_{\skidist \sim \skidistoverdist}\left[\frac{\pialg(\skidist)}{\skibayopt(\skidist)}\right].$$ 
It is useful to observe that for any distribution over distributions, any best response never stops on seeing a bad day of weather. Any strategy which stops on seeing a bad day of weather is strictly dominated by a strategy which stops at the next good day. It is without loss to consider such algorithms. 

An optimal algorithm makes decisions based on the count of the number of good days in a sub-sequence of days. We formally define what we mean by such an information-symmetric algorithm for the ski-rental problem. In \Cref{d:infosymmetricalgs}, the first argument for function $A_{i}$ is the count of good weather days observed in the first $(i-1)$ days. The second argument is the type of weather on the $i^{\text{th}}$ day.

\begin{definition}
\label{d:infosymmetricalgs}
An information-symmetric deterministic algorithm $\pialg = (\alg{1},\alg{2},\dots,\alg{\inplen})$ is such that

$$\alg{i} : \{0,1,\dots,\timeindex-1\} \times \{0,1\} \to \{\stp,\continue\} \quad \forall \quad \timeindex \in [\inplen] $$

\noindent
where $\stp$ is a decision to stop and $\continue$ is a decision to continue.  
\end{definition}

We first prove that for any fixed distribution that the adversary plays, an optimal algorithm, which is obtained by backwards induction, is information-symmetric. 

\begin{lemma}
\label{l:bestresponseinfosymmetric}
For any distribution over distributions, $\skidistoverdist \in \Delta([\delta,1])$, that the adversary plays, an optimal algorithm for that distribution obtained by Backwards Induction is information-symmetric.
\end{lemma}

An implication of this result is that for each sub-sequence of days, it is only the count of the number of good days in that sub-sequence which matters. Utilizing this, we now write our algorithm as a function of the count of good days in a given sub-sequence of days, rather than that of the exact sequence observed. Formally, we define below the class of algorithms we can restrict ourselves to and those which we focus on.

A straightforward application of \Cref{l:reducedgameapxeqlb} to \Cref{l:bestresponseinfosymmetric} tells us that it is without loss to restrict the class of algorithms to those that are information symmetric and do not stop on observing bad weather i.e.,

$$\pialgclass = \{\pialg : \pialg \quad \text{is information-symmetric; does not stop on observing bad weather}\}$$

We show that the four required properties are satisfied for the above class of algorithms. From now on, whenever we talk about an algorithm, we assume that it comes from the above class. The above restriction of information-symmetric algorithms allows us to show that for a discrete distribution played by the adversary, an optimal algorithm can be found in efficient time by Backwards Induction, thus satisfying Efficient Best Response.

\begin{theorem}
\label{t:polybackwardsinduction}
For a discrete distribution with support size $\suppsize{\skidistoverdist}$, Backwards Induction finds an optimal solution to the optimization problem
$$\min_{\pialg \in \pialgclass}\mathbf{E}_{\skidist \sim \skidistoverdist}\left[\frac{\pialg(\skidist)}{\skibayopt(\skidist)}\right],$$
in time $O(\suppsize{\skidistoverdist} T^{3})$. 
\end{theorem}

We now show that the prior-independent ski-rental problem satisfies Efficient Utility Computation. A key observation is our restriction to information-symmetric algorithms defined above. 

\begin{lemma}
\label{l:skirentproptwo}
For any information-symmetric algorithm $\pialg$, and any $\skidist \in [\delta,1]$, the utility $\util(\pialg,\skidist) = \frac{\pialg(\skidist)}{\skibayopt(\skidist)}$ can be computed in time $\poly(\inplen)$.
\end{lemma}

We now show that Small-Cover is satisfied. We show an $\epsilon$-cover of the probability space which satisfies Small-Cover. We then apply the learning framework described in \autoref{s:framework} to show that a prior-independent optimal algorithm can be found in time $\poly\left(T,\frac{1}{\epsilon}\right)$. The idea to obtain an $\epsilon$-cover such that utilities are close to each other is to discretize the grid $[\delta,1]$. An intuition as to why discretizing works is that the expected performance of an algorithm does not change much if we consider two probabilities that are close to each other.

\begin{lemma}
\label{l:algoperfcloseness}
For a fixed algorithm $\pialg$, for any two probabilities $\skidist_{\alpha},\skidist_{\beta}$ such that 
   $$|\skidist_{\alpha}-\skidist_{\beta}| \leq \frac{\epsilon}{\inplen^{2}},$$
\noindent   
the expected costs of the algorithm are such that
   $$|\pialg(\skidist_{\alpha})-\pialg(\skidist_{\beta})| \leq \frac{\epsilon(\stopcost+\inplen)}{\inplen}.$$
\end{lemma}

While \Cref{l:algoperfcloseness} showed an algorithm performance closeness property, what we are really interested is in a utility-closeness property, i.e., for probabilities that are close to each other, the utility $\frac{\pialg(\skidist)}{\skibayopt(\skidist)}$ should be close to each other for any algorithm. The trivial lower bound of $\inplen \delta$ on the Bayesian optimal cost will give a utility-closeness bound as a function of $1/\delta$. This, however is not desirable as $\delta$ could be arbitrarily small. We prove a tighter utility-closeness bound which is independent of $\delta$. Proving the ratio of utilities are close requires an additional decomposition property of the algorithmic cost. The proof of this theorem is provided in \autoref{a:proofs-framework}. 

\begin{lemma}
\label{l:utilitycloseness}
For any $k \in [\min\{\ceil{\inplen-\stopcost},\inplen-1\}]$, consider two probabilites $\skidist_{\alpha},\skidist_{\beta}$ in the interval $\left(\frac{\stopcost-1}{\inplen-k},\min\left\{\frac{\stopcost-1}{\inplen-k-1},1\right\}\right]$ such that 
    $$|\skidist_{\alpha}-\skidist_{\beta}| \leq \frac{\epsilon}{8(\stopcost+1)^{2}(\inplen+1)^{2}\inplen}.$$
    
    \noindent
    Then for any algorithm $\pialg$ which does not stop on observing bad weather, we have
    $$|\util(\pialg,\skidist_{\alpha})-u(\pialg,\skidist_{\beta})| \leq \frac{\epsilon}{4}.$$
    \noindent
    This utility-closeness property also holds when $\skidist_{\alpha},\skidist_{\beta}$ are in the interval $\left[\delta,\frac{\stopcost-1}{\inplen-1}\right]$. 
\end{lemma}

We focus on showing the \textit{Small-Cover} property and show the \textit{Bounded Best Response Utility} property during the proof of \textit{Small-Cover}. It is important to observe that the discretization should be done in a disjoint way for each interval in which the Bayesian optimal algorithm is the same. \Cref{l:utilitycloseness} tells us that the difference in the utility of an algorithm is close only for probabilities in intervals where the Bayesian optimal algorithm is the same. For two probabilities that are close but go across these intervals, the Lipschitz bound we prove will not hold.  

\begin{theorem}
\label{t:epscoverskirental}
The class of algorithms $\pialgclass$ and the class of distributions $\skidistclass$ in the prior-independent ski-rental problem satisfy Small-Cover.
\end{theorem}

Having shown that all four properties in the framework hold, we give a bound on the running time to compute an approximate-optimal prior-independent ski-rental algorithm. 

\begin{theorem}
\label{t:skirentalfptas}
There exists an algorithm to find an $\epsilon$-approximate prior-independent optimal algorithm in time $\Tilde{O}(\inplen^{12}\epsilon^{-3})$. 
\end{theorem}

\section{Empirical Comparison of Models}
\label{s:simulations}

In this section we empirically compare the worst-case algorithm, the subgame optimal algorithm and the prior-independent optimal algorithm in prior-independent framework. The setting of our simulations is as follows. We consider the finite-horizon ski-rental problem. We fix the time horizon at $\inplen=9$. The \textit{stopping cost} $B$ is varied continuously in the range $[1.1,\inplen]$. 

\begin{figure}
    \centering
    \includegraphics[scale=0.45]{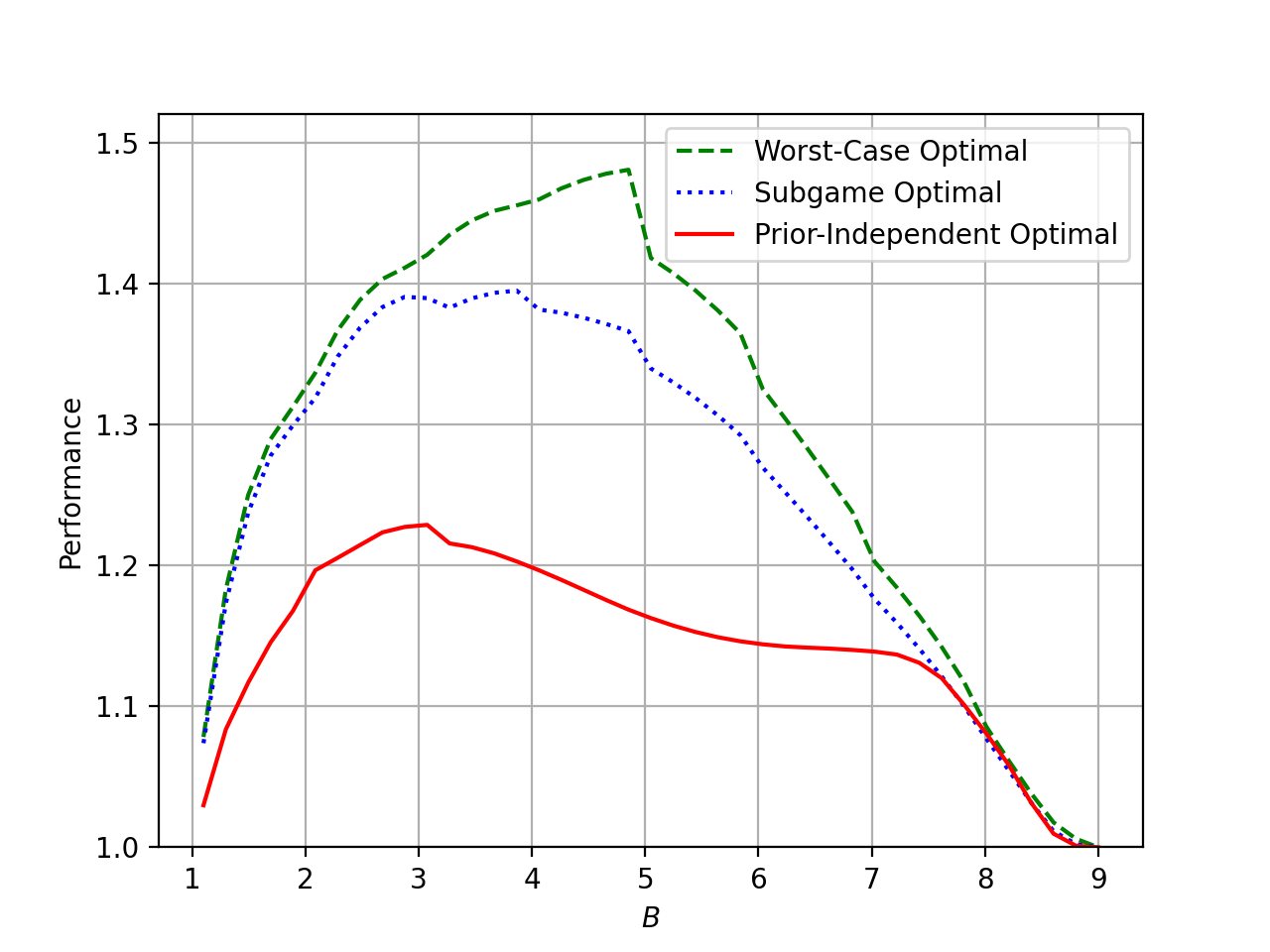}
    \caption{The time horizon is fixed at $T=9$. The \textit{stopping} cost is varied continuously from $[1,T]$. The figure plots the approximation ratio of each algorithm against the worst-case distribution of the adversary in the prior-independent framework.}
    \label{fig:comp_approximationratio_sec5}
\end{figure}


We compare each algorithm in the prior-independent framework, considered in \autoref{s:framework}. We plot the approximation ratio of the algorithm on the worst-case distribution of the adversary. As was done in the previous figure, the plot is as a function of the \textit{stopping cost}. Recall, a distribution of the adversary in the prior-independent setting is a distribution over probabilities. The plot shows the leverage that subgame optimal algorithm have over the worst-case algorithms. Due to statistical assumptions in the prior-independent setting, we expect to see every input sequence with non-zero probability (apart from the case $p=1$). The subgame optimal algorithm is designed to perform optimally (in the worst-case framework) on subgames that are not worst-case. The comparison in \autoref{fig:comp_approximationratio_sec5} depicts that. 

The algorithms were evaluated on the worst-case adversary distributions. A different approach could have compared the performance of each algorithm in either framework. Such a comparison would not depict the advantages a subgame optimal algorithm has over a worst-case algorithm. Their performance in the prior-independent framework of the ski-rental problem, given by the expressions \begin{equation*}
    \max_{p}\frac{\text{WC-OPT}(p)}{\text{OPT}_{p}(p)} \quad \text{and } \max_{p}\frac{\text{SG-OPT}(p)}{\text{OPT}_{p}(p)}
\end{equation*} respectively, is the same. Moreover, it is exactly equal to their performance in the worst-case framework with optimal hindsight cost as its benchmark. We give a formal proof for why this is the case.

Consider the worst-case framework with optimal hindsight cost as its benchmark. The following sets \begin{itemize}
    \item $\mathbf{X}^{\text{WC-OPT}} := \argmax_{\mathbf{X}} \frac{\text{WC-OPT}(\mathbf{X})}{\text{OPT}_{\mathbf{X}}(\mathbf{X})},$
    \item $\mathbf{X}^{\text{SG-OPT}} := \argmax_{\mathbf{X}} \frac{\text{SG-OPT}(\mathbf{X})}{\text{OPT}_{\mathbf{X}}(\mathbf{X})},$ 
\end{itemize} comprise of worst-case input sequences for the worst-case optimal and subgame optimal algorithm respectively. A useful lemma, which we prove, tells us that both these sets have a non-empty intersection.

\begin{lemma}
\label{l:commonworstcaseinput}
    In the worst-case framework with optimal hindsight cost as its benchmark, it holds that, \begin{equation*}
        \mathbf{X}^{\text{WC-OPT}} \cap \mathbf{X}^{\text{SG-OPT}} \neq \emptyset,
    \end{equation*} that is, the worst-case input sequences for the worst-case optimal algorithm and the sub-game optimal algorithm have a non-empty intersection.
\end{lemma}

The goal is to achieve a comparison between the performance of algorithms in the prior-independent and worst-case framework. An important lemma, which we prove below, tells us that the performance of any algorithm in the prior-independent framework is upper bounded by its performance in the worst-case framework. Intuitively, the result makes sense: informational assumptions on the input space can only improve performance. Note however, the benchmark changes from optimal hindsight cost to the Bayesian optimal cost.

\begin{lemma}
\label{l:algpiworstcaseub}
    For any algorithm $A$, it holds, \begin{equation*}
        \max_{F \in \mathcal{F}}\frac{A(F)}{\text{OPT}_{F}(F)} \leq \max_{\mathbf{X}}\frac{A(\mathbf{X})}{\text{OPT}_{\mathbf{X}}(\mathbf{X})},
    \end{equation*} that is, its performance in the prior-independent framework with Bayesian optimal cost as its benchmark is at most its performance in the worst-case framework with the optimal hindsight cost as its benchmark.
\end{lemma}

\begin{proof}
    Let $\beta := \max_{\mathbf{X}}\frac{A(\mathbf{X})}{\text{OPT}_{\mathbf{X}}(\mathbf{X})}$ be the performance of algorithm $A$ in the worst-case framework. Let $F^* \in \argmax_{F}\frac{A(F)}{\text{OPT}_{F}(F)}$ be a worst-case distribution for algorithm $A$ in the prior-independent framework. It holds that, \begin{align*}
        \frac{A(F^*)}{OPT_{F^*}(F^*)} &= \frac{\mathbf{E}_{\mathbf{X} \sim F^*}[A(\mathbf{X})]}{\mathbf{E}_{\mathbf{X} \sim F^*}[OPT_{F^*}(\mathbf{X})]}, \\
        & \leq \frac{\beta \mathbf{E}_{\mathbf{X} \sim F^*}[\text{OPT}_{\mathbf{X}}(\mathbf{X})]}{\mathbf{E}_{\mathbf{X} \sim F^*}[OPT_{F^*}(\mathbf{X})]}, \\
        & \leq \beta,
    \end{align*} where the first equality follows by definition. The first inequality follows from $\beta$ being the performance of the algorithm in the worst-case framework. The final inequality follows from observing that $\text{OPT}_{\mathbf{X}}(\mathbf{X})$ minimizes the hindsight cost of the algorithm for $\mathbf{X}$. Thus, $\text{OPT}_{\mathbf{X}}(\mathbf{X}) \leq \text{OPT}_{F^*}(\mathbf{X})$ for any input sequence $\mathbf{X}$. 
\end{proof}

Finally, we state a result which solves for the performance of the subgame optimal and worst-case optimal algorithm in the prior-independent framework. We provide sufficient conditions on the class of distributions $\mathcal{F}$ such that the performance of these algorithms can be characterized. The sufficient condition covers the following scenario: \begin{enumerate}
    \item There exists a distribution $F$ such that all the input sequences $\mathbf{X}$ with non-zero mass in $F$ are worst-case.
    \item The cost of the Bayesian optimal algorithm for distribution $F$ on every input sequence $\mathbf{X}$ in its support is exactly equal to the optimal hindsight cost on $\mathbf{X}$.
\end{enumerate} An environment in which this condition is satisfied is when an input is generated by a point mass distribution such that it is worst-case. The result below is applied to prior-independent ski-rental problem. 

\begin{lemma}
\label{l:worstcaseperfinpi}
    If there exists a distribution $F \in \mathcal{F}$ such that \begin{itemize}
        \item $\text{supp} (F) \subseteq \mathbf{X}^{\text{WC-OPT}} \cap \mathbf{X}^{\text{SG-OPT}}$,
        \item $\text{OPT}_{\mathbf{X}}(\mathbf{X}) = \text{OPT}_{F}(\mathbf{X})$ for all $\mathbf{X} \in \text{supp}(F)$,
    \end{itemize} the performance of \text{WC-OPT} and \text{SG-OPT} in the prior-independent framework is the same. Moreover, it is equal to their performance in the worst-case framework with optimal hindsight cost as their benchmark.
\end{lemma}

\begin{proof}
Let $A$ denote either of WC-OPT or SG-OPT. By assumption, their performance on distribution $F$ in the prior-independent framework can be solved as \begin{align*}
    \frac{A(F)}{\text{OPT}_{F}(F)} &= \frac{\mathbf{E}_{\mathbf{X} \sim F}[A(\mathbf{X})]}{\mathbf{E}_{\mathbf{X} \sim F}[\text{OPT}_{F}(\mathbf{X})]}, \\
    &=  \frac{\beta \mathbf{E}_{\mathbf{X} \sim F}[\text{OPT}_{\mathbf{X}}(\mathbf{X})]}{\mathbf{E}_{\mathbf{X} \sim F}[\text{OPT}_{F}(\mathbf{X})]}, \\
    &= \beta.
\end{align*}
The first equality is by definition of the performance of algorithm $A$. The second equality by assumption that $\text{supp}(F) \subseteq \mathbf{X}^{\text{WC-OPT}} \cap \mathbf{X}^{\text{SG-OPT}}$. The final equality by assumption that $\text{OPT}_{\mathbf{X}}(\mathbf{X}) = \text{OPT}_{F}(\mathbf{X})$ for all $\mathbf{X} \in \text{supp}(F)$. Since distribution $F$ achieves the upper bound, as proven in Lemma \ref{l:algpiworstcaseub}, the statement follows. 
\end{proof}

An application of the result above shows that the performance of worst-case optimal and subgame optimal algorithm is the same in the prior-independent ski-rental framework.  The distribution $p=1$ satisfies the conditions of Lemma \ref{l:worstcaseperfinpi}.

\begin{corollary}
    The performance of the worst-case optimal ski-rental algorithm and subgame optimal ski-rental algorithm is the same in the prior-independent ski-rental framework. Moreover, it is equal to their performance in the worst-case framework with optimal hindsight cost as its benchmark.
\end{corollary}

\section{Conclusion}
\label{s:conclusion}

In this paper we study the design and analysis of online algorithms, considering notions in a beyond worst-case analysis. We take a game-theoretic approach to the analysis of such problems. 

First, we introduce the concept of {\em subgame optimal} online algorithms. This is a stronger solution concept than worst-case algorithms. Subgame optimal algorithms are required to be optimal at every stage of the process of the input being revealed. We solve for the subgame optimal algorithm in the finite-horizon ski-rental problem. We show that difference in the performance of a worst-case algorithm and the subgame optimal algorithm is as bad as the worst-case approximation ratio.

Second, we define prior-independent optimality for online algorithms. This is a weaker requirement than traditional competitive analysis. We provide sufficient conditions on the online algorithm design problem which imply the existence of a FPTAS. We show that these properties are satisfied in the prior-independent finite-horizon ski-rental problem.

Finally, we empirically compare the different solution concepts in prior-independent framework.

\bibliographystyle{apalike}
\bibliography{references}

\appendix
\section{Preliminaries}
\label{s:appendixprelims}

\subsection{Online Algorithm Design}
\label{p:onlinealgs}

An online algorithm design problem involves making decisions over time with an uncertain future. At each stage, a decision maker observes an input. Upon observing this, she makes an irrevocable decision without knowledge of future inputs. An input comes from the class $\mathcal{X}$. An input sequence of length $\inplen$ is denoted by $\mathbf{X}$, where $\mathbf{X} \in \mathcal{X}^{\inplen}$. 

Before observing an input, the decision maker is at a state $\mathcal{S}$ in the system. Once the input is observed, she alters the state of the system at some cost. A mapping from inputs to a decision to a change in the state of the system is known as an algorithm. Formally,  

\begin{definition}
\label{d:bralg}
A deterministic algorithm is a tuple of functions $(\alg{1},\alg{2},\dots,\alg{\inplen})$ where

$$\alg{i} : \left(\inpelement{1},\dots,\inpelement{\timeindex}\right) \to \mathcal{S} \quad \forall \left(\inpelement{1},\dots,\inpelement{\timeindex}\right) \in \mathcal{X}^{\timeindex} \quad, \forall \timeindex \in [\inplen].$$
\end{definition}

The algorithm suffers a cost for moving between states. For initial state $s'$, final state $s''$ and input $X_{i}$, the cost suffered is $c(s',s'',X_{i})$, where $c : \mathcal{S} \times \mathcal{S} \times \mathcal{X} \to \mathbb{R}^{+}$. The system starts at initial state $s_{0}$ before observing any input.  The cost suffered by an algorithm on an input sequence $\mathbf{X}$, denoted by $A(s_{0},\mathbf{X})$, is $$A(s_{0},\mathbf{X}) = c(s_{0},A_{1}(X_{1}),X_{1}) + \sum_{t=2}^{\inplen}c\left(A_{t-1}(X_{1},\dots,X_{t-1}),A_{t}(X_{1},\dots,X_{t}),X_{t}\right).$$ 


\subsection{Zero-Sum Games}

Consider a two-player zero-sum game between an algorithm and an adversary. Let the pure strategy space of the algorithm be $\pialgclass$ and the pure strategy space of the adversary be $\pidistclass$. For this preliminary section, let the cardinality of both these sets be finite. The utility to the adversary for a pure strategy profile $(\pialg,\pidist)$ is finite and denoted as $\util(\pialg,\pidist)$. Two-player zero-sum games have the property that all equilibria give a player the same expected payoff, and consequently the other player the negated payoff. This is called the value of the game. Moreover, in finite games, each player has a mixed strategy that guarantee this value for all strategies of the other player.  

\begin{fact}[\citealp{vonNeumann-28}]
In any Nash equilibrium of a two-player zero-sum game, the expected payoff obtained is the same and is equal to the value of the game.
\end{fact}

The mixed strategy that guarantees this value for the algorithm is known as a Minmax strategy while the mixed strategy that guarantees this value for the adversary is known as a Maxmin strategy. 

\begin{definition}
The Minmax strategy of the algorithm solves
$$\argmin_{\pirandalg \in \Delta(\pialgclass)}\max_{\pidist \in \pidistclass} \mathbf{E}_{\pialg \sim \pirandalg}[u(\pialg,\pidist)].$$
\end{definition}

\begin{definition}
The Maxmin strategy of the adversary solves
$$\argmax_{\pidistoverdist \in \Delta(\pidistclass)}\min_{\pialg \in \pialgclass} \mathbf{E}_{\pidist \sim \pidistoverdist}[u(\pialg,\pidist)].$$
\end{definition}

The values achieved by each of these strategies is the same and is equal to the value of the game. This is known as the Minimax Theorem. 

\begin{fact}[\citealp{vonNeumann-28}]
In a zero-sum game,

$$ \min_{\pirandalg \in \Delta(\pialgclass)}\max_{\pidist \in \pidistclass} \mathbf{E}_{\pialg \sim \pirandalg}[u(\pialg,\pidist)] = \max_{\pidistoverdist \in \Delta(\pidistclass)}\min_{\pialg \in \pialgclass} \mathbf{E}_{\pidist \sim \pidistoverdist}[u(\pialg,\pidist)].$$

\end{fact}

\begin{fact}[\citealp{vonNeumann-28}]
In any Nash equilibrium of the zero-sum game, the adversary plays a maxmin strategy while the algorithm plays a minmax strategy.
\end{fact}

An approximate Nash equilibrium can be defined as follows.

\begin{definition}($\epsilon$-Approximate Nash Equilibrium)
In a two-player zero-sum game, a strategy profile $(\pirandalg,\pidistoverdist)$ is an $\epsilon$-approximate Nash equilibrium if
$$\mathbf{E}_{\pialg \sim \pirandalg}\left[\mathbf{E}_{\pidist \sim \pidistoverdist}\left[\util(\pialg,\pidist)\right]\right] \geq \max_{\pidist \in \pidistclass}\mathbf{E}_{\pialg \sim \pirandalg}\left[\util(\pialg,\pidist)\right] - \epsilon,$$
and 
$$\mathbf{E}_{\pialg \sim \pirandalg}\left[\mathbf{E}_{\pidist \sim \pidistoverdist}\left[\util(\pialg,\pidist)\right]\right] \leq \min_{\pialg \in \pialgclass}\mathbf{E}_{\pidist \sim \pidistoverdist}\left[\util(\pialg,\pidist)\right] + \epsilon.$$
\end{definition}

An analogous fact to the minimax theorem holds for approximate Nash equilibrium and approximate minimax strategies. We provide a proof for this fact in \autoref{a:proofs-prelims}.

\begin{fact}
\label{f:apxnashminmax}
In any zero-sum game, if $(\pirandalg,\pidistoverdist)$ is an $\epsilon$-approximate Nash equilibrium, then $\pirandalg$ is an $2\epsilon$-approximate Minmax strategy and $\pidistoverdist$ is an $2\epsilon$-approximate Maxmin strategy. 
\end{fact}

\section{Preliminaries on Computing Approximate Nash Equilibria in Zero-Sum Games}
\label{a:proofs-prelims}

We describe two approaches, at a high level, to compute an approximate Nash equilibrium.

In a no-regret learning setting, a learner is equipped with a set of actions. Each action provides the learner with a corresponding utility. The utility given by an action varies over a time horizon. The goal of the learner is to pick the action with the best hindsight utility, i.e., the action with maximum total utility computed over time. In the no-regret learning setting, the learner is unaware of future utilities that an action could provide. Based on the history of utilities up until each time instant, a no-regret learning algorithm chooses a distribution over its actions. The payoff obtained by the learner at a time instant is the expected utility over actions. 

The goal is to come up with a no-regret learning algorithm which minimizes regret; where regret is the difference between the optimal hindsight utility and total expected utility obtained by the algorithm. a no-regret learning algorithm is said to be no-regret if, for any sequence of utilities, the regret converges to zero in the limit of the time horizon. Furthermore, such an algorithm is said to be efficient no-regret if at each time instant, the probability distribution over actions is updated in time polynomial in the number of actions.

We use an efficient no-regret policy to compute an approximate equilibrium in a half-infinite two-player zero-sum game using a common technique in which the player with a finite set of strategies employs an efficient no-regret strategy while the opposing player utilizes a best response oracle to best respond to the mixed strategy of the opposing player at each iteration of the algorithm. We describe the exact algorithm along with the exact guarantee in \autoref{a:proofs-prelims}.

An approximate Nash equilibrium in a two-player zero-sum game can be found using the Ellipsoid method. The key idea to finding such an equilibrium is that the problem can be reduced to solving a linear program (\citealp{Dantzig-63}). To solve for the mixed strategy of a player in equilibrium, the set of variables in the linear program correspond to its mixed strategy and the value that it can guarantee itself in the game. The constraints correspond to strategies of the adversary player. \citet{HLP-19} show that if a player has access to a best-response oracle, and the opposing player has a finite set of actions, then the ellipsoid method can be used to compute an equilibrium in time polynomial in the size of the finite strategy space and the time taken by the best-response oracle. At a high level, the best-response oracle acts as a separation oracle by testing whether the player can guarantee their proposed value for its current mixed strategy.

\subsection{Online Learning}

In an online learning setting, there is an learner with a set of actions $\{1,\dots,\numactions\}$ and chooses an action at each time $\timestep \in \{1,2,\dots,\inplen\}$. The learner faces a sequence of payoffs $(\boldsymbol{\util}^{\timestep})_{\timestep=1}^{\timehorizon}$, where $\boldsymbol{\util}^{\timestep} \in [0,\Bar{u}]^{\numactions}$ corresponds to the utility obtained by each action at time $\timestep$. The learner maintains a probability distribution $\boldsymbol{\probdist}^{\timestep}$ over its action set for each time $\timestep \in [\timehorizon]$ given they have observed the sequence of utility vectors $(\boldsymbol{\util}^{\subtimestep})_{\subtimestep=1}^{\timestep-1}$. They obtain the corresponding expected payoff $\boldsymbol{\probdist}^{\timestep} \cdot \boldsymbol{\util}^{\timestep}$ at time $\timestep$. The agent wants to minimize its regret.

\begin{definition}(Regret) 
The difference between the utility of the best action in hindsight and the cumulative utility obtained by the agent

$$\regret(\timehorizon) = \max_{i \in [\numactions]}\sum_{\timestep=1}^{\timehorizon}\util^{\timestep}_{i} - \sum_{\timestep=1}^{\timehorizon}\boldsymbol{\probdist}^{\timestep} \cdot \boldsymbol{\util}^{\timestep}$$

\end{definition}

The goal of the learner is to come up with a policy, i.e., a mapping from the sequence of utilities to a probability distribution over actions $\boldsymbol{\probdist}^{\timestep}$, such that its regret is sub-linear in $\timehorizon$ for any sequence of utilities. Such a policy is known as a no-regret policy. We reduce our prior-independent framework to an online learning setting. To that end, we define a policy to be an efficient no-regret policy if it is a no-regret policy that needs to observe a small number of utility vectors to achieve vanishing regret and if it updates its probability distribution over actions in a computationally efficient manner. Formally defining such a policy,

\begin{definition}(Efficient No-Regret Policy) A policy is an efficient no-regret policy if 

\begin{itemize}
    \item For any sequence of utility functions $(\boldsymbol{\util}^{\timestep})_{\timestep=1}^{\timehorizon}$, the regret $R(T) \leq O(T^{\alpha})$ for some $\alpha \in (0,1)$ 

    \item It achieves $\epsilon$-regret after observing a sequence of $\poly\left(\log m, \frac{1}{\epsilon},\Bar{u}\right)$ utility vectors
    \item It updates the probability distribution over actions $\mathbf{\probdist}^{\timestep}$ in $\poly(m)$ time. 
\end{itemize}
\label{d:efficientnoregret}
\end{definition}

\noindent
An example of an efficient no-regret policy is Multiplicative Weights (\citealp{FS-99}). 

An important part of our work is the use of online learning algorithms to compute an approximate Nash-equilibrium in zero-sum games. The utility to the adversary for a strategy profile $(\pialg,\pidist)$ is $\util (\pialg,\pidist)$. Assume that the utilities lie in $[0,\Bar{u}]$. 

Considering the learning protocol given below where Player 2 employs an efficient no-regret policy and Player 1 best responds.

\begin{algorithm}

\SetAlgoLined 

\textbf{Input:}

$\pidistclass$ : Action set of Player 2 \;
$\epsilon$ : Desired regret guarantee

\textbf{Initialization:} 

$(\pidistoverdist)^{1}$ : Initial mixed strategy of Player 2 over actions $\pidistclass$. \;
$T$ : Time horizon given by the Efficient No-Regret Policy after which Player 2 is guaranteed to have at most $\epsilon$ regret.

\textbf{Algorithm:}

\For{$\timestep \gets 1$ \KwTo $\timehorizon$}{
  Player 1 best responds to $(\pidistoverdist)^{\timestep}$ with algorithm $\pialg^{\timestep} \in \pialgclass$. We overload notation and refer to $\pialg^{\timestep}$ as a degenerate mixed strategy on algorithm $\pialg^{\timestep}$\;
 Player 2 observes the vector of utilities $\boldsymbol{\util}^{\timestep}$ \;
 Player 2 updates its mixed strategy to $(\pidistoverdist)^{\timeindex+1}$ according to an efficient no-regret policy (\Cref{d:efficientnoregret}). 
}
    
\caption{Online Learning in Zero-Sum Games by \citealp{FS-99}}
\label{a:prelimlearningalgo}
\end{algorithm}

\begin{lemma}
\label{l:appxnashlearn}
If after $\timehorizon$ iterations, the regret of Player 2 is $\frac{R(T)}{T} \leq \epsilon$, then $\frac{1}{\timehorizon}\sum_{\timestep=1}^{\timehorizon}(\pidistoverdist)^{\timestep}$ is an $\epsilon$-approximate Maxmin strategy and $\frac{1}{\timehorizon}\sum_{\timestep=1}^{\timehorizon}\pialg^{\timestep}$ is an $\epsilon$-approximate Minmax strategy and the strategy profile $\left(\frac{1}{\timehorizon}\sum_{\timestep=1}^{\timehorizon}\pialg^{\timestep}, \frac{1}{\timehorizon}\sum_{\timestep=1}^{\timehorizon}(\pidistoverdist)^{\timestep} \right)$ is an $\epsilon$-approximate Nash equilibrium. 
\end{lemma}

\begin{proof}

Define $\avgprobdist = \frac{1}{\timehorizon}\sum_{\timestep=1}^{\timehorizon}(\pidistoverdist)^{\timestep}$ and $\avgalgdist = \frac{1}{\timehorizon}\sum_{\timestep=1}^{\timehorizon}\pialg^{\timestep}$. Consider the following chain of inequalities, 

\begin{align*}
    \max_{\pidistoverdist \in \Delta(\pidistclass)}\min_{\pialg \in \pialgclass}\mathbf{E}_{\pidist \sim \pidistoverdist}[\util(\pialg,\pidist)] & \geq \min_{\pialg \in \pialgclass}\mathbf{E}_{\pidist \sim \avgprobdist}[\util(\pialg,\pidist)] \tag*{(1)} \label{pf2:eq1} \\
    & = \min_{\pialg \in \pialgclass}\frac{1}{\timehorizon}\sum_{\timestep}^{\timehorizon}\util(\pialg,(\pidistoverdist)^{\timestep}) \tag*{(2)} \label{pf2:eq2} \\
    & \geq \frac{1}{\timehorizon} \sum_{\timestep}^{\timehorizon}\min_{\pialg \in \pialgclass}\util(\pialg,(\pidistoverdist)^{\timestep}) \tag*{(3)} \label{pf2:eq3} \\
    & = \frac{1}{\timehorizon} \sum_{\timestep}^{\timehorizon}\util(\pialg^{\timestep},(\pidistoverdist)^{\timestep}) \tag*{(4)} \label{pf2:eq4} \\
    & \geq \max_{\pidist \sim \pidistclass}\frac{1}{\timehorizon}\sum_{\timestep=1}^{\timehorizon}\util(\pialg^{\timestep},\pidist) - \epsilon \tag*{(5)} \label{pf2:eq5} \\
    & = \max_{\pidist \sim \pidistclass}\mathbf{E}_{\pialg \sim \avgalgdist}[\util(\pialg,\pidist)] - \epsilon \tag*{(6)} \label{pf2:eq6} \\
    & \geq \min_{\pirandalg \in \Delta(\pialgclass)}\max_{\pidist \in \pidistclass}\mathbf{E}_{\pialg \sim \pirandalg}[\util(\pialg,\pidist)] - \epsilon \tag*{(7)} \label{pf2:eq7}
\end{align*}

\noindent
where the first inequality follows from the fact that $\avgprobdist$ is a feasible randomization, the first equality from writing out the expectation explicitly, the second inequality from the fact that point-wise minimizing each term will perform better for the algorithm, the second equality from noting that at each iteration $\pialg^{\timestep}$ was a best response to $(\pidistoverdist)^{\timestep}$, the third inequality from the condition that the regret of Player 2 is at most $\epsilon$, the third equality from writing down the expectation explicitly and the fourth inequality from observing that $\avgalgdist$ is a feasible randomized algorithm. Comparing equation \ref{pf2:eq6} with the LHS of equation \ref{pf2:eq1} shows that $\avgalgdist$ is an $\epsilon$-approximate Minmax strategy. Similarly, comparing equation \ref{pf2:eq1} with equation \ref{pf2:eq7} shows that $\avgprobdist$ is an $\epsilon$-approximate Maxmin strategy. 

To show that $(\avgalgdist,\avgprobdist)$ is an $\epsilon$-approximate Nash equilibrium. Consider the following chain of inequalities. 

$$\mathbf{E}_{\pialg \sim \avgalgdist}[\mathbf{E}_{\pidist \sim \avgprobdist}[\util(\pialg,\pidist)]] \geq \min_{\pialg \in \pialgclass}\mathbf{E}_{\pidist \sim \avgprobdist}[\util(\pialg,\pidist)] \geq \max_{\pidist \in \pidistclass}\mathbf{E}_{\pialg \sim \avgalgdist}[\util(\pialg,\pidist)] - \epsilon$$

\noindent
where the second inequality follows from comparing equation \ref{pf2:eq1} and \ref{pf2:eq6}. Also consider

$$\mathbf{E}_{\pialg \sim \avgalgdist}[\mathbf{E}_{\pidist \sim \avgprobdist}[\util(\pialg,\pidist)]] \leq \max_{\pidist \in \pidistclass}\mathbf{E}_{\pialg \sim \avgalgdist}[\util(\pialg,\pidist)] \leq \min_{\pialg \in \pialgclass}\mathbf{E}_{\pidist \sim \avgprobdist}[\util(\pialg,\pidist)] + \epsilon$$

where the second inequality again follows from comparing equation \ref{pf2:eq1} and \ref{pf2:eq6}. Thus $(\avgalgdist,\avgprobdist)$ is an $\epsilon$-approximate Nash equilbrium. 
\end{proof}

\subsection{Ellipsoid Method}
\label{ss:eqlbellipsoid}

In this section we show that the Ellipsoid method can be used to compute an approximate Nash equilibrium in a two-player zero-sum game where a player who has uncountably many pure strategies has access to a best response oracle. The problem of finding an equilibrium strategy for each player can be written as a linear program where the set of variables map to a mixed strategy for a player and each constraint corresponds to a pure strategy of the opposing player. Unfortunately, uncountable number of strategies for a player correspond to uncountable number of variables and thus cannot be solved efficiently. However, an equilibrium can still be computed by solving the dual program and considering a reduced game where the set of uncountable strategies is reduced to those that were violated during the run of Ellipsoid algorithm for the dual program. This argument was proved by \citet{HLP-19}. We provide it here for completeness. 

\begin{proof}[Proof of \autoref{l:prelimlearningexist}]
    Let the setting of the zero-sum game be such that \textit{Player 1} has an uncountably large pure strategy space denoted by $\pialgclass$. Keeping convention, the utility of \textit{Player 2} for a strategy profile $(\pialg,\pidist)$, where $\pialg \in \pialgclass$ and $\pidist \in \pidistclass$, is $\util(\pialg,\pidist)$. The linear program which solves for the mixed strategy of \textit{Player 2} in equilibrium is given by  \begin{equation*} \begin{aligned}
    & \max_{\left(x_{\pidist}\right)_{\pidist \in \pidistclass},z} && z \\
    & \text{subject to} && \sum_{\pidist \in \pidistclass}x_{\pidist}\util(\pialg,\pidist) \geq z, &&  \forall \pialg \in \pialgclass, \\
    & && \sum_{\pidist \in \pidistclass}x_{\pidist} = 1, &&  \\
    & && x_{\pidist} \geq 0, && \forall \pidist \in \pidistclass.
  \end{aligned}
\end{equation*}
This linear program can be solved using the Ellipsoid method, for a suitably chosen initial point which will be described later, where the separation oracle used during the course of algorithm is the best response oracle of \textit{Player 1}. Let $\pidist_{e}$ be some fixed pure strategy of \textit{Player 2}. The linear program to solve for the mixed strategy in equilibrium can be rewritten in the standard form as \begin{equation} \label{eqn:advlp}  \begin{aligned}
    & \max_{\left(x_{\pidist}\right)_{\pidist \in \pidistclass \setminus \pidist_{e}},z} && z \\
    & \text{subject to} && \sum_{\pidist \in \pidistclass \setminus \pidist_{e}}x_{\pidist}[\util(\pialg,\pidist_{e})-\util(\pialg,\pidist)] + z \leq \util(\pialg,\pidist_{e}) , &&  \forall \pialg \in \pialgclass, \\
    & && \sum_{\pidist \in \pidistclass \setminus \pidist_{e}}x_{\pidist} \leq 1, &&  \\
    & && x_{\pidist} \geq 0, && \forall \pidist \in \pidistclass \setminus \pidist_{e}.
  \end{aligned}
  \tag{\textit{ADV-LP}}
\end{equation} Let the variables at iteration $t$ of the algorithm be denoted as $\left(x_{\pidist}^{t}\right)_{\pidist \in \pidistclass \setminus \pidist_{e}},z^{t}$. The oracle best-responds to the mixed strategy of \textit{Player 2} given by the variables of the program to obtain $\left(\pialg^{*}\right)^{t}$. The feasibility of the point is then verified by computing the expected utility at the best response. In the case $$\sum_{\pidist \in \pidistclass \setminus F_{e}}x_{\pidist}^{t}\util\left(\left(\pialg^{*}\right)^{t},\pidist\right) + \left(1-\sum_{\pidist \in \pidistclass \setminus F_{e}}x_{\pidist}^{t}\right)\util\left(\left(\pialg^{*}\right)^{t},\pidist_{e}\right) < z^{t},$$ the point is infeasible and a separating hyperplane is given by the vector $$\left([\util(\left(\pialg^{*}\right)^{t},\pidist)-\util(\left(\pialg^{*}\right)^{t},\pidist_{e})]_{\pidist \in \pidistclass \setminus F_{e}},-1\right).$$ In the case $$\sum_{\pidist \in \pidistclass \setminus F_{e}}x_{\pidist}^{t}\util\left(\left(\pialg^{*}\right)^{t},\pidist\right) + \left(1-\sum_{\pidist \in \pidistclass \setminus F_{e}}x_{\pidist}^{t}\right)\util\left(\left(\pialg^{*}\right)^{t},\pidist_{e}\right) \geq z^{t},$$, the additional constraints we need to check is whether the variables $\left(x_{\pidist}\right)_{\pidist \in \pidistclass \setminus \pidist_{e}}$ are non-negative and whether $\sum_{\pidist \in \pidistclass \setminus \pidist_{e}}x_{\pidist} \leq 1$. If any of the variables is non-negative, the separating hyperplane is the unit vector of dimension $|\pidistclass|$ with the coordinate corresponding to the violated constraint being $1$, while if the constraint $\sum_{\pidist \in \pidistclass \setminus \pidist_{e}}x_{\pidist} \leq 1$ is violated, the separating hyperplane is the vector $\left(1,1,\dots,1,0\right)$ with dimension $|\pidistclass|$. The total time taken by the separation oracle is linear in the time taken by the separation oracle and $|\pidistclass|$.

While running Ellipsoid to solve the program described above will give us the mixed strategy of \textit{Player 2} in equilibrium, it says nothing about the mixed strategy of \textit{Player 1}. To solve for the mixed strategy of \textit{Player 1}, we would have to solve the dual program of \ref{eqn:advlp} given by \begin{equation*}  \begin{aligned}
    & \max_{\left(y_{\pialg}\right)_{\pialg \in \pialgclass},z} && z \\
    & \text{subject to} && -\sum_{\pialg \in \pialgclass}y_{\pialg}\util(\pialg,\pidist) \geq z, &&  \forall \pidist \in \pidistclass, \\
    & && \sum_{\pialg \in \pialgclass}y_{\pialg} = 1, &&  \\
    & && y_{\pialg} \geq 0, && \forall \pialg \in \pialgclass.
  \end{aligned}
\end{equation*}
This program would not be feasible to solve using Ellipsoid as the number of variables, where a variable corresponds to a pure strategy of \textit{Player 1}, is uncountably large. We now show that considering a restricted strategy space suffices to find an approximate minmax strategy of \textit{Player 1}.

Consider the set of algorithms corresponding to the violated constraints in the run of Ellipsoid method while computing an $\epsilon$ approximate solution for the program \ref{eqn:advlp}. Let $\pialgclass_{v}$ be the set of algorithms corresponding to the violated constraints in solving \ref{eqn:advlp} and consider the reduced game where the strategy space of \textit{Player 2} is $\pidistclass$ and the strategy space of \textit{Player 1} is $\pialgclass_{v}$. Crucially, observe that the size of the pure strategy space of \textit{Player 1} in the reduced game is polynomial in $\text{poly}\left(|\pidistclass|,\log\frac{1}{\epsilon}\right)$ The linear program to solve for the mixed strategy of \textit{Player 1} in this reduced game is given by \begin{equation*}\label{eqn:agtrlp}\begin{aligned}
    & \max_{\left(y_{\pialg}\right)_{\pialg \in \pialgclass_{v}},z} && z \\
    & \text{subject to} && -\sum_{\pialg \in \pialgclass_{v}}y_{\pialg}\util(\pialg,\pidist) \geq z, &&  \forall \pidist \in \pidistclass, \\
    & && \sum_{\pialg \in \pialgclass_{v}}y_{\pialg} = 1, &&  \\
    & && y_{\pialg} \geq 0, && \forall \pialg \in \pialgclass_{v}.
  \end{aligned}
  \tag{\textit{ALGR-LP}}
\end{equation*}
The dual of this program allows us to solve for the mixed strategy of \textit{Player 2} in the reduced game, which when written in the standard form, is given by \begin{equation} \label{eqn:advrlp}  \begin{aligned}
    & \max_{\left(x_{\pidist}\right)_{\pidist \in \pidistclass \setminus \pidist_{e}},z} && z \\
    & \text{subject to} && \sum_{\pidist \in \pidistclass \setminus \pidist_{e}}x_{\pidist}[\util(\pialg,\pidist_{e})-\util(\pialg,\pidist)] + z \leq \util(\pialg,\pidist_{e}) , &&  \forall \pialg \in \pialgclass_{v}, \\
    & && \sum_{\pidist \in \pidistclass \setminus \pidist_{e}}x_{\pidist} \leq 1, &&  \\
    & && x_{\pidist} \geq 0, && \forall \pidist \in \pidistclass \setminus \pidist_{e}.
  \end{aligned}
  \tag{\textit{ADVR-LP}}
\end{equation}
Crucially, the sequence of points generated during the run of Ellipsoid algorithm to solve \ref{eqn:advrlp} is identical to the sequence generated during the run of Ellipsoid to solve \ref{eqn:advlp}. A consequence of this is that the variable $z$ converges to the same variable value in both programs. Thus, if we solve for an $\epsilon$-approximate maxmin strategy, in the game with strategy spaces $\left(\pialgclass,\pidistclass\right)$, by solving \ref{eqn:advlp}, then the difference between the value of the game with pure strategy spaces $\left(\pialgclass_{v},\pidistclass\right)$ and the value of the game with pure strategy spaces $\left(\pialgclass,\pidistclass\right)$ is at most $\epsilon$. Solving for an $\epsilon$-approximate minmax strategy in the reduced game, by solving \ref{eqn:agtrlp} using the Ellipsoid method with a separation oracle that goes over each constraint in a brute-force manner, then implies a $2\epsilon$-approximate minmax strategy in the original game. 
\end{proof}
\section{Proofs For Section 3}
\label{a:sub-game-proofs}

\begin{proof}[Proof of \Cref{l:truerandworstcase}]
To prove this, we fix a mixed Nash equilibrium in \textit{SRP-R} and show that neither player has an incentive to deviate when considering their strategy profiles in \textit{SRP}.

Consider the mixed strategy of the algorithm player in this equilibrium. The algorithms in \textit{SRP-R} do not depend on the location of the bad weather days. Thus, any order of a fixed number of good weather days gives the adversary the same payoff. It follows that in \textit{SRP}, the adversary has no incentive to deviate.

Consider the mixed strategy of the adversary player in this equilibrium. The algorithm player only needs to take a decision of when to \textit{stop}. The equilibrium strategy of the algorithm player in \textit{SRP-R} does that optimally. It follows that in \textit{SRP}, the algorithm player has no incentive to deviate. 
Since no player wants to strictly deviate from the mixed Nash equilibrium strategy profile of \textit{SRP-R}, the statement follows.
\end{proof}

\begin{proof}[Proof of \autoref{t:randworstcase}]
We find the mixed strategy of the algorithm player in (\textit{SRP-R}). Consider a mixed strategy for the algorithm player, defined as $\{\beta_{0},\dots,\beta_{\inplen-\floor{B}-1},\beta_{\inplen}\}$, where $\beta_{l}$ is the probability that the algorithm plays pure strategy $A^{l}$. For such a mixed strategy, the expected payoff that the adversary player obtains for pure strategy $F^{k}$ is $$\arraycolsep=1.4pt\def\arraystretch{2.2}\mathbf{E}[u(A^{l},F^{k})] = \threepartdef{\sum\limits_{m=0}^{k-1}\beta_{m}\frac{B+m}{k} + \sum\limits_{m=k}^{\inplen - \floor{B} - 1}\beta_{m} + \beta_{\inplen}}{1 \leq k \leq \inplen - \floor{B} - 1}{\sum\limits_{m=0}^{\inplen - \floor{B} - 1}\beta_{m}\frac{B+m}{k} + \beta_{\inplen}}{\inplen - \floor{B} \leq k \leq \ceil{B}-1}{\sum\limits_{m=0}^{\inplen - \floor{B} - 1}\beta_{m}\frac{B+m}{B} + \frac{\inplen}{B}\beta_{\inplen}}{k = \inplen}.$$ Our analysis can be divided into two cases.   \begin{itemize}
    \item \textbf{Case 1} : $\ceil{\stopcost} + \floor{\stopcost} \geq \inplen + 1$. Observe that the set of pure strategies $\{F^{k} : k \in \{\inplen - \floor{B} + 1, \dots, \ceil{B} - 1\}\}$ is weakly dominated by $F^{\inplen - \floor{B}}$. Assuming $k \in \{\inplen - \floor{B},\dots,\ceil{B}-1\}$, for any algorithm $A^{l}$ where $l \in \{0,\dots,\inplen-\floor{B}-1\}$, the utility to the adversary is $\frac{B+l}{k}$. For algorithm $A^{\inplen}$, the utility to the adversary is $1$. It follows that $F^{\inplen - \floor{B}}$ weakly dominates, as it gives a weakly higher payoff for every algorithm.
    Thus, it is without loss to further reduce the pure strategy space of the adversary to $$\mathcal{F}_{g} = \{F^{k} : k \in \{1,\dots,\inplen - \floor{B}\} \cup \{\inplen\}\}.$$ We now solve for mixed Nash equilibrium strategies by assuming that the adversary player puts non-zero mass on every pure strategy in the set described above. By assumption, the expected utility for every pure strategy in $\mathcal{F}_{g}$ is equal. We solve this system of equations by induction, considering a pair of expected utility equations at each stage, to derive values for $\{\beta_{1},\dots,\beta_{\inplen-\floor{B}-1}\}$ in terms of $\beta_{0}$. For the base case, consider the expected utility equations corresponding to strategy $F_{1}$ and $F_{2}$. It follows from equating the two that \begin{align*}
        \beta_{0}B + \sum_{m=1}^{\inplen-\floor{B}-1}\beta_{m} + \beta_{T} &= \beta_{0}\frac{B}{2} + \beta_{1}\frac{B+1}{2} + \sum_{m=2}^{\inplen-\floor{B}-1}\beta_{m} + \beta_{T} \\
        \iff \beta_{1} &= \frac{B}{B-1}\beta_{0}.
    \end{align*} Assume that for $l \in \{1,2,\dots,\inplen-\floor{B}-2\}$, $$\beta_{l} = \left(\frac{B}{B-1}\right)^{l}\beta_{0}.$$ To prove the induction hypothesis for stage $(l+1)$, consider the expected utility equations for strategies $F_{l+1}$ and $F_{l+2}$. Equating them, we have that \begin{align*}
        \sum\limits_{m=0}^{l+1}\beta_{m}\frac{B+m}{l+2} + \sum\limits_{m=l+2}^{\inplen - \floor{B} - 1}\beta_{m} + \beta_{\inplen} &=  \sum\limits_{m=0}^{l}\beta_{m}\frac{B+m}{l+1} + \sum\limits_{m=l+1}^{\inplen - \floor{B} - 1}\beta_{m} + \beta_{\inplen},\\
        \iff \left(\frac{B+l+1}{l+2}-1\right)\beta_{l+1} &= \sum_{m=0}^{l}(B+m)\left(\frac{1}{l+1}-\frac{1}{l+2}\right)\beta_{m}, \\
        \iff (l+1)(B-1)\beta_{l+1} &= \sum_{m=0}^{l}(B+m)\beta_{m} \\
        \iff (l+1)(B-1)\beta_{l+1} &= \beta_{0}\sum_{m=0}^{l}(B+m)\left(\frac{B}{B-1}\right)^{m}, \\
        \iff \beta_{l+1} &= \left(\frac{B}{B-1}\right)^{l+1}\beta_{0}.
    \end{align*}
    We now solve a pair of simultaneous equations in $\beta_{0}$ and $\beta_{T}$ to arrive at a closed form expression for the mixed Nash equilibrium. First, we must have that \begin{equation}
    \label{eq:simulone}
        \begin{aligned}
        \sum_{m=0}^{\inplen-\floor{B}-1}\beta_{m} + \beta_{T} &= 1, \\
        \iff (B-1)\left[\left(\frac{B}{B-1}\right)^{\inplen-\floor{B}}-1\right]\beta_{0} + \beta_{T} &= 1, 
    \end{aligned}
    \end{equation} where the first equality follows from the fact that we are considering a mixed strategy of the algorithm player. The the second equality follows from the relation between $\beta_{0}$ and $\beta_{l}$, for $l \in \{1,\dots,\inplen-\floor{B}-1\}$, derived above. We obtain a second equation by equating the expected utility function of strategy $F_{1}$ and $F_{T}$. It holds that \begin{equation}
    \label{eq:simultwo}
        \begin{aligned}
            \sum_{m=0}^{\inplen-\floor{B}-1}\beta_{m}\frac{B+m}{B} + \frac{T}{B}\beta_{T} &= B\beta_{0} + \sum_{m=1}^{\inplen-\floor{B}-1}\beta_{m} + \beta_{T}, \\
            \iff \beta_{0}\left(\sum_{m=0}^{\inplen-\floor{B}-1}\frac{B+m}{B}\left(\frac{B}{B-1}\right)^{m}\right) + \frac{T}{B}\beta_{T} &= B\beta_{0} + \beta_{0}\left(\sum_{m=1}^{\inplen-\floor{B}-1}\left(\frac{B}{B-1}\right)^{m}\right) + \beta_{T}, \\
            \iff \frac{\inplen-B}{B}\beta_{\inplen} &= (B+\floor{B}-\inplen)\left(\frac{B}{B-1}\right)^{\inplen-\floor{B}-1} \beta_{0}.
        \end{aligned}
    \end{equation} Combining \autoref{eq:simulone} and \autoref{eq:simultwo}, we have that $$\beta_{0} = \frac{1}{\frac{B\floor{B}}{\inplen-B}\left(\frac{B}{B-1}\right)^{\inplen-\floor{B}-1} - (B-1)}.$$ The result follows by replacing $\beta_{0}$ in \autoref{eq:simultwo} to obtain $\beta_{\inplen}$. 
    \item \textbf{Case 2} : $\ceil{\stopcost} + \floor{\stopcost} < \inplen + 1$. We follow the same calculation as Case 1 to arrive at an equilibrium strategy. Observe that the set of strategies $\{A^{l} : l \in \{\ceil{B},\dots,\inplen-\floor{B}-1\} \cup \{\inplen\}\}$ are weakly dominated by $A^{\ceil{B}-1}$. Assuming $l \geq \ceil{B}-1$, for any $F^{k}$ where $k \leq \ceil{B}-1$, the utility to the adversary is $1$. For input $F^{\inplen}$, the utility to the adversary for $A^{l}$, where $l \in \{\ceil{B}-1,\dots,\inplen-\floor{B}-1\}$, is $\frac{B+l}{B}$ and for $A^{\inplen}$ is $\frac{\inplen}{B}$. By assumption that $\ceil{B} + \floor{B} < \inplen + 1$, it follows that $$\frac{\ceil{B}+B-1}{B}< \frac{\ceil{B} + \floor{B}}{B} \leq \frac{\inplen}{B}.$$ Thus, $A^{\ceil{B}-1}$ weakly dominates $A^{l}$ for any $l \geq \ceil{B}$. Consider a mixed strategy of the algorithm player defined as $\{\beta_{0},\dots,\beta_{\ceil{B}-1}\}$. By following the same calculation as was done in Case 1, assuming that the adversary player puts mass on every pure strategy in equilibrium, we have that $$\beta_{l} = \left(\frac{B}{B-1}\right)^{l}\beta_{0}, \quad \forall l \in \{1,\dots,\ceil{B}-2\}.$$ To solve for $\beta_{0}$ and $\beta_{\ceil{B}-1}$, we first have 
    \begin{equation}
        \label{eq:simulthree}
        \begin{aligned}
        \sum_{m=0}^{\ceil{B}-2}\beta_{m} + \beta_{\ceil{B}-1} &= 1, \\
        \iff (B-1)\left[\left(\frac{B}{B-1}\right)^{\ceil{B}-1}-1\right]\beta_{0} + \beta_{\ceil{B}-1} &= 1, 
    \end{aligned}
    \end{equation} where the first equality follows from the fact that we are considering a mixed strategy of the algorithm player. By assumption of the adversary mixed strategy having full support, we equate the expected utilities for pure strategies $F_{\inplen}$ and $F_{1}$ to obtain \begin{equation}
        \begin{aligned}
            \sum_{m=0}^{\ceil{B}-1}\beta_{m}\frac{B+m}{B} &= B\beta_{0} + \sum_{m=1}^{\ceil{B}-1}\beta_{m}, \\
            \iff \sum_{m=0}^{\ceil{B}-2}\left(\frac{B}{B-1}\right)^{m}\frac{m}{B}\beta_{0} + \frac{\ceil{B}-1}{B}\beta_{\ceil{B}-1} &= (B-1)\beta_{0}, \\
            \iff \left(\frac{B}{B-1}\right)^{\ceil{B}-2}\left(B+1-\ceil{B}\right)\frac{B}{\ceil{B}-1}\beta_{0} &= \beta_{\ceil{B}-1}.
        \end{aligned}
        \label{eq:simulfour}
    \end{equation} Solving \autoref{eq:simulthree} and \autoref{eq:simulfour}, we have that $$ \pushQED{\qed} \beta_{0} = \frac{1}{\left(B-1\right)\left[\left(\frac{B}{B-1}\right)^{\ceil{B}-1}\frac{B}{\ceil{B}-1}-1\right]}. $$  
\end{itemize} The theorem statement then follows.
\end{proof}

\begin{proof}[Proof of \Cref{l:reducetimehorizon}]
    It is weakly optimal to always \textit{continue} on a bad weather weather day. There are $0$ costs incurred for continuing on the initial $l$ contiguous bad weather days. There is no contribution to the optimal hindsight cost from these bad weather days. Thus, it is an instance of the finite-horizon problem with horizon length $(T-l)$ on the $(l+1)^{\text{st}}$ day.
\end{proof}

    \begin{proof}[Proof of \Cref{l:characterizesubgame}]
Conditional on \textit{continuing} until $(X_{1},\dots,X_{i},1)$, the minimization problem to the algorithm player is the same as that of observing a sequence where all the bad weather days arrive first in a contiguous manner followed by good weather days. By \Cref{l:reducetimehorizon}, the probability of \textit{stopping}, conditional on \textit{continuing} for the observed sequence is given by $\eta_{k}^{\inplen-i+k-1}$.
\end{proof}
\section{Proofs For Section 4}
\label{a:proofs-framework}

We present the missing proofs from \autoref{s:framework}.

\begin{proof}[Proof of \autoref{l:frameworkdisctocont}]
Consider a strategy $\pidistoverdist$ of the adversary in the original game. We propose a corresponding strategy $\pidistoverdisteps$ of the adversary in the reduced game such that for any algorithm $\pialg$, their payoffs are $\epsilon$-close to one another. Index the elements in $\pidistclassepsgrid$ as $\pidistepsindex \in [|\pidistclassepsgrid|]$ and let the $\pidistepsindex^{th}$ element in the order be defined as $(\pidisteps)_{\pidistepsindex}$. Denote the $\epscoverscaling \epsilon$-ball of distributions around $\pidisteps$ as $B_{\epscoverscaling \epsilon}(\pidisteps)$. For $\pidistepsindex=1$, make the assignment $Pr[\pidistoverdisteps = (\pidisteps)_{1}] = Pr[\pidistoverdist \in B_{\epscoverscaling \epsilon}((\pidisteps)_{1}) \cap \pidistclass]$. Consider the set $\pidistclassrem{1} = \pidistclass \setminus B_{\epscoverscaling\epsilon}((\pidisteps)_{1})$. For $\pidistepsindex=2$, make the assignment $Pr[\pidistoverdisteps = (\pidisteps)_{2}] = Pr[\pidistoverdist \in B_{\epscoverscaling\epsilon}((\pidisteps)_{2}) \cap \pidistclassrem{1}]$. If the set $B_{\epscoverscaling\epsilon}((\pidisteps)_{2}) \cap \pidistclassrem{1}$ is a null set, assign the probability to be zero.  Consider the set $\pidistclassrem{2} = \pidistclassrem{1} \setminus B_{\epscoverscaling\epsilon}((\pidisteps)_{2})$. For a general $\pidistepsindex$, make the assignment $Pr[\pidistoverdisteps = (\pidisteps)_{\pidistepsindex}] = Pr[\pidistoverdist \in B_{\epscoverscaling\epsilon}((\pidisteps)_{\pidistepsindex}) \cap \pidistclassrem{\pidistepsindex-1}]$. If the set $B_{\epscoverscaling\epsilon}((\pidisteps)_{\pidistepsindex}) \cap \pidistclassrem{\pidistepsindex-1}$ is a null set, assign the probability to be zero. For the next iteration, consider the set $\pidistclassrem{\pidistepsindex} = \pidistclassrem{\pidistepsindex-1} \setminus B_{\epscoverscaling\epsilon}((\pidisteps)_{\pidistepsindex})$. Since $\pidistclassepsgrid$ is an $\epsilon$-cover and $\pidistoverdist$ is a feasible distribution, this process will terminate and $\pidistoverdisteps$ will be a feasible distribution on $\pidistclassepsgrid$. We now show that for any algorithm, the payoffs under these two distributions are $O(\epsilon)$-close to one another. For ease of notation, we re-write $\util(\pialg,\pidist) := \frac{\pialg(\pidist)}{\bayopt(\pidist)}$. We can write

\begin{equation*}
    \begin{split}
        \mathbf{E}_{\pidist \sim \pidistoverdist}[\util(\pialg,\pidist)] & = \int_{\pidist}\util(\pialg,\pidist)d\pidistoverdist, \\
        & = \sum_{\pidistepsindex=0}^{|\pidistclassepsgrid|}\int_{B_{\epscoverscaling\epsilon}((\pidisteps)_{\pidistepsindex}) \cap \pidistclass_{\pidistepsindex-1}}\util(\pialg,\pidist)dF',
    \end{split}
\end{equation*}

\begin{equation*}
    \begin{split}
        & \leq \sum_{\pidistepsindex=0}^{|\pidistclassepsgrid|}\left(\util(\pialg,(\pidisteps)_{\pidistepsindex}) + \frac{\epsilon}{4}\right)\int_{B_{\epscoverscaling\epsilon}((\pidisteps)_{\pidistepsindex}) \cap \pidistclass_{\pidistepsindex-1}}dF', \\
        & = \mathbf{E}_{\pidisteps \sim \pidistoverdisteps}[\util(\pialg,\pidisteps)] + \frac{\epsilon}{4},
    \end{split}
\end{equation*}
\noindent
where the first equality follows from the definition of the expectation, the second equality from the fact that $\pidistclassepsgrid$ is an $\epscoverscaling \epsilon$-cover and the sets that are the limits of the integral are disjoint, the first inequality from \Cref{p:frameworkutilcloseness}. The same calculation can be done for the lower bound i.e., 
$$ \mathbf{E}_{\pidist \sim \pidistoverdist}[\util(\pialg,\pidist)] \geq \mathbf{E}_{\pidisteps \sim \pidistoverdisteps}[\util(\pialg,\pidisteps)] - \frac{\epsilon}{4}.$$
\end{proof}

\begin{proof}[Proof of \autoref{l:frameworkepsapprox}]
Consider an $\frac{\epsilon}{4}$-approximate Nash equilibrium $(\pirandalgepseq,\pidistoverdistepseq)$ in the reduced game. Consider a best response $\pidistoverdist$ to $\pirandalgepseq$ in the original game. Consider the corresponding strategy to $\pidistoverdist$ in the reduced game as generated in \autoref{l:frameworkdisctocont} and denote it as $\pidistoverdisteps$. Since $(\pirandalgepseq,\pidistoverdistepseq)$ were an $\epsilon$-Nash equilibrium,
        
        \begin{equation*}
            \begin{split}
                \mathbf{E}_{\pialg \sim \pirandalgepseq}[\mathbf{E}_{\pidisteps \sim \pidistoverdistepseq}[\util(\pialg,\pidisteps)]] & \geq \max_{\pidisteps \in \pidistclassepsgrid}\mathbf{E}_{\pialg \sim \pirandalgepseq}[\util(\pialg,\pidisteps)] - \frac{\epsilon}{4} \\ & \geq \mathbf{E}_{\pialg \sim \pirandalgepseq}[\mathbf{E}_{\pidisteps \sim \pidistoverdisteps}[\util(\pialg,\pidisteps)] - \frac{\epsilon}{4}\\
                & \geq \mathbf{E}_{\pialg \sim \pirandalgepseq}[\mathbf{E}_{\pidist \sim \pidistoverdist}[\util(\pialg,\pidist)] - \frac{\epsilon}{2} \\
            \end{split}
        \end{equation*}
\noindent        
where the first inequality follows from the fact that $(\pirandalgepseq,\pidistoverdistepseq)$ is an $\epsilon$-approximate Nash equilibrium in the reduced game. The second inequality from noting that $\pidistoverdisteps$ is a feasible distribution in the reduced game. The third equality follows from \autoref{l:frameworkdisctocont}. It follows that $(\pirandalgepseq,\pidistoverdistepseq)$ is an $\frac{\epsilon}{2}$-approximate Nash equilibrium in the original game as $\pidistoverdist$ is a best response to $\pirandalgepseq$ in the original game and the algorithm has no incentive to $\epsilon$-deviate as its own strategy space was not restricted i.e., 
        $$\mathbf{E}_{\pialg \sim \pirandalgepseq}[\mathbf{E}_{\pidisteps \sim \pidistoverdistepseq}[\util(\pialg,\pidisteps)]] \leq \min_{\pialg \in \pialgclass}\mathbf{E}_{\pidisteps \sim \pidistoverdistepseq}[\util(\pialg,\pidisteps)] + \frac{\epsilon}{4} $$
\noindent
which follows from the definition that $(\pirandalgepseq,\pidistoverdistepseq)$ was an $\frac{\epsilon}{4}$-Nash equilibrium.
\end{proof}

\begin{proof}[Proof of \autoref{l:reducedgameapxeqlb}]
    Let $\left(\pialg_{r}',\pidist'\right)$ be an $\epsilon$-approximate mixed Nash equilibrium in the reduced game. By definition of an approximate Nash equilibrium, we have that \begin{align}
    \label{e:epseqlb}
        \mathbf{E}_{\pidist \sim \pidist'}\left[\mathbf{E}_{\pialg \sim \pialg_{r}'}\left[u\left(\pialg,\pidist\right)\right]\right] \geq \mathbf{E}_{\pidist \sim \pidist'}\left[u\left(\pialg_{r},\pidist\right)\right] - \epsilon, \quad \forall \pialg_{r} \in \pialgclass^{*}.
    \end{align} Let there exist a pure strategy in $\pialgclass$ which gives player 1 an $\epsilon$-incentive to deviate, i.e., for some $\pialg \in \pialgclass$, \begin{align*}
        \mathbf{E}_{\pidist \sim \pidist'}\left[u\left(\pialg,\pidist\right)\right] > \mathbf{E}_{\pidist \sim \pidist'}\left[\mathbf{E}_{\pialg \sim \pialg_{r}'}\left[u\left(\pialg,\pidist\right)\right]\right] + \epsilon.
    \end{align*} Since there exists a pure strategy in the original game which provides an $\epsilon$-incentive to deviate, any best response to $\pidist'$ in the original game will also provide player 1 an $\epsilon$-incentive to deviate. By assumption that for any mixed strategy that player $2$ chooses, a best response of player $1$ lies in a subset of strategies $\pialgclass^{*}$, there exists $\pialg_{r} \in \pialgclass^{*}$ such that \begin{align*}
        \mathbf{E}_{\pidist \sim \pidist'}\left[u\left(\pialg_{r},\pidist\right)\right] > \mathbf{E}_{\pidist \sim \pidist'}\left[\mathbf{E}_{\pialg \sim \pialg_{r}'}\left[u\left(\pialg,\pidist\right)\right]\right] + \epsilon.
    \end{align*} This contradicts \autoref{e:epseqlb}. Thus, player 1 does not have $\epsilon$-incentive to deviate in the original game. 

    Player 2 does not have $\epsilon$-incentive to deviate either. This is because we did not restrict the strategy space of player 2 while considering the reduced game. The statement follows.
\end{proof}

\begin{proof}[Proof of \autoref{t:bayopt}]
For $\skidist \in \left(\frac{\stopcost-1}{\inplen-k},\frac{\stopcost-1}{\inplen-k-1}\right]$, where $k \in [\min\{\ceil{\inplen-\stopcost},\inplen-1\}]$, the algorithm will always continue on seeing a good day for all days $\timeindex \geq (k+1)$. This is because at the last day, the algorithm always continues, no matter if the weather was good or bad. Let the agent continue on seeing good days at time indices $\timeindex \geq \timeindex'$, where $\timeindex' \geq (k+2)$. Consider the decision the algorithm makes at time $\timeindex'-1$ on seeing a good day. The agent will also continue as the stopping cost is strictly greater than the continuation costs, i.e $\stopcost \geq 1+\skidist (\inplen-k-1) > 1+\skidist (\inplen-\timeindex'+1)$, where $\timeindex' \geq (k+2)$. 

Consider the decision the algorithm makes on seeing a good day at time index $k$. The cost of stopping is $\stopcost$, while the cost of continuing is $1+\skidist (\inplen-k)$. Since, $\skidist > \frac{\stopcost-1}{\inplen-k}$, we have that the agent will stop on seeing a good day at time index $k$. The optimal algorithm also stops, when it observes a good day, for $\timeindex \leq k$. We prove this by induction. The base case is already argued for above. Let the algorithm stop at time index $k' \leq k$ on seeing a good day. Consider the decision at time $(k'-1)$. The cost from stopping is $\stopcost$. The cost, if it were to continue can be written as
\begin{equation*}
    \begin{split}
        \contcost{k'-1}{1} & = 1 + \skidist \stopcost + (1-\skidist)\contcost{k'}{0} \\
        & = 1 +\skidist \stopcost + (1-\skidist)[\skidist \contcost{k'+1}{1}+(1-\skidist)\contcost{k'+1}{0}]
    \end{split}
\end{equation*}
where the second equality follows the optimal algorithm always continuing on seeing bad weather. As the algorithm stopped on seeing a good day at time index $k'$, it must have been that $\stopcost < 1+\skidist \contcost{k'+1}{1}+(1-\skidist)\contcost{k'+1}{0}$. But then this implies that the agent must stop at time $k'-1$ when it sees a good day as
\begin{equation*}
    \begin{split}
        \stopcost & < 1+\skidist \contcost{k'+1}{1}+(1-\skidist)\contcost{k'+1}{0} \\
        & < \frac{1}{1-\skidist} + \skidist \contcost{k'+1}{1}+(1-\skidist)\contcost{k'+1}{0}
    \end{split}
\end{equation*}
and thus $\stopcost < 1 +\skidist \stopcost + (1-\skidist)[\skidist \contcost{k'+1}{1}+(1-\skidist)\contcost{k'+1}{0}]$. 
For $\skidist \leq \frac{\stopcost-1}{\inplen-1}$, the agent will always continue as at any time index the cost of stopping is greater than the cost of continuing given that the algorithm always continues at greater time indices. 
\end{proof}

We now give a proof for $\autoref{l:bestresponseinfosymmetric}$.

\begin{proof}[Proof of \autoref{l:bestresponseinfosymmetric}] For a distribution over probabilities $\skidistoverdist$, an optimal algorithm is obtained by backwards induction. The sketch of the proof is to use induction and then use the fact that for the class of distributions we consider the Bayesian updated distribution will be the same for a same count of good days in any sub-sequence of days. 

Observe that the agent never stops on seeing a bad day of weather and thus the property holds whenever the agent observes $\inpelement{\timeindex} = 0$ for all $\timeindex \in [\inplen]$. Now consider the case the agent observes $\inpelement{\timeindex}=1$. Note that at the last time-step, no matter what the sequence of days observed, the agent will always continue i.e.

$$\alg{\inplen}(\inpelement{1},\dots,\inpelement{\inplen}) = \continue \quad \forall (\inpelement{1},\dots,\inpelement{\inplen}) \in \{0,1\}^{\inplen}$$

This is because of our normalization that the stopping cost $\stopcost$ is strictly greater than the continuation cost of one. Thus the information-symmetry property is satisfied for functions at stage $T$. Let this property be satisfied by functions above stage $(i+1)$. At stage $i$, for a particular sub-sequence of days $(\inpelement{1},\dots,\inpelement{\timeindex-1},1)$, backwards induction will make a comparison between the cost of stopping the sequence and the cost of continuing. The agent will Bayesian update given the knowledge of $(\inpelement{1},\dots,\inpelement{\timeindex-1},1)$. Observe that for sub-sequences such that $\sum_{l=1}^{\timeindex}\inpelement{l} = k$ for any $k \in \{0\} \cup [\timeindex-1]$, the continuation cost is the same because the decisions taken in stages $(i+1)$ and further on are the same due to the information-symmetry property of functions above stage $(i+1)$. By the assumption on our class of distributions, we have that the Bayesian updated distribution will also be the same given sub-sequences of days with the same count of good days. Hence the same decision will be taken by the backwards induction algorithm and thus information-symmetry property is satisfied by functions at stage $\timeindex$.  
\end{proof}

\begin{proof}[Proof of \autoref{t:polybackwardsinduction}]
For a fixed distribution, an optimal algorithm can be found by backwards induction. Recall that \autoref{l:bestresponseinfosymmetric} allows us to focus on those algorithms which are information-symmetric i.e we can only focus on the count of the number of goods for a sub-sequence of a particular length and those which only make stop/continue decisions on observing a good day of weather. Let us denote a sub-sequence of length $\timeindex \in [T]$, by $\piinputsubseq{\timeindex} = (\inpelement{1},\dots,\inpelement{\timeindex})$. Let us denote the continuation cost, as a function of the probability $\skidist \in supp(\skidistoverdist)$ for a particular distribution $\skidistoverdist$, of the algorithm obtained by backwards induction, on seeing a sub-sequence $\piinputsubseq{\timeindex}$ as $\contcostp{\sum_{l=1}^{\timeindex-1}\inpelement{\timeindex}}{\timeindex}{\inpelement{\timeindex}}$. Note that for any sub-sequence of length $\inplen$, the agent will continue no matter what i.e. the agent always continues on the last day. This just follows from our normalization that $\stopcost > 1$ and that the cost of continuing is one. Thus we have,
    \begin{align*}
        \contcostp{k}{\inplen}{1} & = 1 \quad & \forall  \quad k \in \{0\} \cup [\inplen-1], \forall \quad \skidist \in supp(\skidistoverdist) \\
        \contcostp{k}{\inplen}{0} & = 0 \quad & \forall   \quad k \in \{0\} \cup [\inplen-1], \forall \quad \skidist \in supp(\skidistoverdist).
    \end{align*}
 For any $\piinputsubseq{\inplen-2} \in \{0,1\}^{\inplen-2}$, to compute the Bayesian updated distribution for a sub-sequence $(\piinputsubseq{\inplen-2},1)$ requires $O(\suppsize \inplen)$ operations. Since we have $O(\inplen)$ such relevant histories (because we focus only on the count of the number of good days), the total operations required to compute the expected continuation payoffs at the last stage are $O(\suppsize \inplen^{2})$. Now consider a sub-sequence of length $\timeindex \in [\inplen-1]$. For a particular count of good days $k \in \{0\} \cup [\timeindex-1]$, to calculate the decision to be taken, which we denote as $\skibayoptdistoverdistcomp{\timeindex}(k,1)$, we need to compare the cost from stopping to the cost of continuing. We need not consider the costs already incurred up until this point as they are sunk. In the equations below, $Pr[(k,\timeindex,1)|\skidist]$ denotes the probability of observing a particular sub-sequence of length $\timeindex$, where among the first $\timeindex-1$ days we observe $k$ good days and the $\timeindex^{th}$ day is good. 
\begin{multline*}
    \skibayoptdistoverdistcomp{\timeindex}(k,1) = arg\,min_{\{\stp,\continue\}} \{\sum_{\skidist \in supp(\skidistoverdist)}Pr[(k,\timeindex,1)|\skidist]Pr[\skidistoverdist=\skidist]\frac{\stopcost}{\skibayopt(\skidist)}, \\ \sum_{\skidist \in supp(\skidistoverdist)}Pr[(k,\timeindex,1)|\skidist]Pr[\skidistoverdist=\skidist]\left[\frac{1 + \skidist \contcostp{k+1}{\timeindex+1}{1} + (1-\skidist)\contcostp{k+1}{\timeindex+1}{0}}{\skibayopt(\skidist)}\right] \}
\end{multline*}
 where the first term is the cost from stopping while the second term is the cost from continuing. In the case $\skibayoptdistoverdistcomp{\timeindex+1}(k,1) = \stp$, then $\contcostp{k}{\timeindex}{1} = \stopcost$, otherwise we have $\contcostp{k}{\timeindex}{1} = 1 + \skidist \contcostp{k+1}{\timeindex+1}{1} + (1-\skidist)\contcostp{k+1}{\timeindex+1}{0} $. We also need to compute the continuation costs if the weather was a bad day. 
 $$\contcostp{k}{\timeindex}{0}  = \skidist \contcostp{k}{\timeindex+1}{1} + (1-\skidist)\contcostp{k}{\timeindex+1}{0} $$
 Computing each of these, i.e $\forall k \in \{0\} \cup [\timeindex]$, requires $O(\suppsize{\skidistoverdist} \timeindex^{2})$ time.  Thus an optimal algorithm for a fixed distribution with support size $\suppsize{\skidistoverdist}$ can be found in $O(\suppsize{\skidistoverdist} T^{3})$. 
\end{proof} 

\begin{proof}[Proof of \autoref{l:skirentproptwo}]
The expected cost of any information symmetric algorithm can be computed by a dynamic program. Let us denote the expected continuation cost from time index $\timeindex$, on observing $\inpelement{\timeindex}$ and on observing a count of $k$ good days in the first $(\timeindex-1)$ indices as  $\contcostp{k}{\timeindex}{\inpelement{\timeindex}}$. We can focus on continuation costs of this form as we are working with information-symmetric algorithms. At the last time step we have that

\begin{align*}
        \contcostp{k}{\inplen}{1} & = \stopcost \cdot \1[\alg{\inplen}(k,1)=\stp] + \1[\alg{\inplen}(k,1)=\continue] \quad & \forall  \quad k \in \{0\} \cup [\inplen-1] \\
        \contcostp{k}{\inplen}{0} & = \stopcost \cdot \1[\alg{\inplen}(k,0)=\stp] \quad & \forall   \quad k \in \{0\} \cup [\inplen-1]
\end{align*}

\noindent
which follow from the definitions of the game i.e. either the algorithm stops and incurs a cost of $\stopcost$ or otherwise it continues and suffers a cost of one or zero depending on whether $\inpelement{\inplen}=1$ or $\inpelement{\inplen}=0$ respectively. At any previous time index $\timeindex \in [\inplen-1]$, we would have that $\forall k \in \{0\} \cup [\timeindex-1]$

\begin{align*}
    \contcostp{k}{\timeindex}{1}  &= \stopcost \cdot \1[\alg{\timeindex}(k,1)=\stp] + \left(1+p\contcostp{k+1}{\timeindex+1}{1} + (1-p)\contcostp{k+1}{\timeindex+1}{0}\right)\1[\alg{\timeindex}(k,1)=\continue] \\
    \contcostp{k}{\timeindex}{0}  & = \stopcost \cdot \1[\alg{\timeindex}(k,0)=\stp] + \left(p\contcostp{k}{\timeindex+1}{1} + (1-p)\contcostp{k}{\timeindex+1}{0}\right)\1[\alg{\timeindex}(k,1)=\continue]
\end{align*}

The expected cost of the information-symmetric algorithm is then

$$\pialg(\skidist) = \skidist \contcostp{0}{1}{1} + (1-\skidist)\contcostp{0}{1}{0}$$

\noindent
Such a cost calculation will take time $O(\inplen^{3})$ to compute. While this was for $\pialg$, note that $\skibayopt$ is also information-symmetric and thus its expected cost can be computed in the same manner which will take time $O(\inplen^{3})$. Thus the utility can be computed in $O(\inplen^{3})$.
\end{proof}

We now provide missing proofs for the small cover property. We first prove a decomposition property of the Bayesian optimal cost, specifically that it can be split into a linear component and a component which is lower bounded by one for any probability in the range where the Bayesian optimal is relevant. 

\begin{proof}[Proof of \autoref{l:algoperfcloseness}]
Consider the following two random processes. One where, where a sequence of days, denoted by $\piinput'$, is sampled i.i.d with probability $\skidist_{\alpha}$, and independently a second random process where a sequence of days, denoted by $\piinput''$, is sampled i.i.d with probability $\skidist_{\beta}$. Consider the random variable $\pialg(\piinput')-\pialg(\piinput'')$, which is nothing but the difference in the algorithms cost on the two sampled sequences. Applying the union bound, these two processes sample the same sequence with probability at least $1-\frac{\epsilon}{\inplen}$, and thus $\pialg(\piinput')-\pialg(\piinput'')=0$ with probability at least $1-\frac{\epsilon}{\inplen}$. The maximum cost that any algorithm can suffer on any sequence of sampled days is $(\stopcost+\inplen)$, as the worst case is when you continue for all days and stop at the end. Thus we have that,
\[ -(\stopcost+\inplen)\frac{\epsilon}{\inplen} \leq \pialg(\skidist_{\alpha})-\pialg(\skidist_{\beta}) \leq (\stopcost+\inplen)\frac{\epsilon}{\inplen}.
\qedhere\] The statement follows.
\end{proof}

\begin{lemma}
\label{l:appendixbayoptdecompose}
For any $k \in [\min\{\ceil{\inplen-\stopcost},\inplen-1\}]$, consider the interval $\skidist \in \left(\frac{\stopcost-1}{\inplen-k},\min\left\{\frac{\stopcost-1}{\inplen-k-1},1\right\}\right]$. The Bayesian optimal cost can be decomposed as
$$\skibayopt(\skidist)= \skidist . \decompfunc(\skidist),$$
where $\decompfunc(\skidist) \geq 1 \quad \forall \skidist \in \left(\frac{\stopcost-1}{\inplen-k},\min\left\{\frac{\stopcost-1}{\inplen-k-1},1\right\}\right]$. This also holds for the Bayesian optimal cost in the interval $\skidist \in \left[\delta,\frac{\stopcost-1}{\inplen-1}\right]$.
\end{lemma}

\begin{proof}
First consider the probability interval $\skidist \in \left[\delta,\frac{\stopcost-1}{\inplen-1}\right]$. The Bayesian optimal algorithm is to continue for all days no matter what the weather and thus $\skibayopt(\skidist)=T\skidist$. The lemma follows for this interval. Now consider any other interval $\skidist \in \left(\frac{\stopcost-1}{\inplen-k},\min\left\{\frac{\stopcost-1}{\inplen-k-1},1\right\}\right]$. The Bayesian optimal cost can be written as
$$\skibayopt(\skidist) =\sum_{\piinput \in \{0,1\}^{\inplen}}\skibayopt(\piinput)\skidist^{\sum_{\timeindex=1}^{\inplen}\inpelement{\timeindex}}(1-\skidist)^{\inplen-\sum_{\timeindex=1}^{\inplen}\inpelement{\timeindex}}.$$ We overloaded notation to denote the cost suffered by the algorithm on a sequence of days as $\skibayopt(\piinput)$. We have that $\skibayopt$ never stops on seeing bad weather and thus $\skibayopt(\{0\}^{T})=0$. Thus $\skibayopt(\skidist)$ can be rewritten as
\begin{align*}
        \skibayopt(\skidist) & = \sum_{\piinput \in \{0,1\}^{\inplen} \backslash \{0\}^{\inplen}}\skibayopt(\piinput)\skidist^{\sum_{\timeindex=1}^{\inplen}\inpelement{\timeindex}}(1-\skidist)^{\inplen-\sum_{\timeindex=1}^{\inplen}\inpelement{\timeindex}} \\
        & = \skidist \left[ \sum_{\piinput \in \{0,1\}^{\inplen} \backslash \{0\}^{\inplen}}\skibayopt(\piinput)\skidist^{\sum_{\timeindex=1}^{\inplen}\inpelement{\timeindex}-1}(1-\skidist)^{\inplen-\sum_{\timeindex=1}^{\inplen}\inpelement{\timeindex}} \right ] \\
        & = \skidist.\decompfunc(\skidist).
\end{align*}
Analyzing $\decompfunc(\skidist)$, we have that

\begin{equation*}
    \begin{split}
        \decompfunc(\skidist) &= \sum_{\piinput \in \{0,1\}^{\inplen} \backslash \{0\}^{\inplen}}\skibayopt(\piinput)\skidist^{\sum_{\timeindex=1}^{\inplen}\inpelement{\timeindex}-1}(1-\skidist)^{\inplen-\sum_{\timeindex=1}^{\inplen}\inpelement{\timeindex}}, \\
        & \geq \sum_{\piinput \in \{0,1\}^{\inplen} \backslash \{0\}^{\inplen}}\skidist^{\sum_{\timeindex=1}^{\inplen}\inpelement{\timeindex}-1}(1-\skidist)^{T-\sum_{\timeindex=1}^{T}\inpelement{\timeindex}}, \\
        & = \sum_{l=1}^{\inplen}\binom{\inplen}{l}\skidist^{l-1}(1-\skidist)^{\inplen-l} = \sum_{r=0}^{\inplen-1}\binom{\inplen}{r+1}\skidist^{r}(1-\skidist)^{\inplen-1-r}, \\
        & = \sum_{r=0}^{\inplen-1}\frac{\inplen}{r+1}\binom{\inplen-1}{r}\skidist^{r}(1-\skidist)^{\inplen-1-r}, \\
        & \geq \sum_{r=0}^{\inplen-1}\binom{\inplen-1}{r}\skidist^{r}(1-\skidist)^{\inplen-1-r} = 1,
    \end{split}
\end{equation*}
where the first inequality follows from that along any path with at least one good day, the algorithm will suffer a cost of at least one as either it will continue and suffer a cost of one or stop and suffer a cost of $B > 1$. 
\end{proof}

Observe that the proof did not specifically use the structure of the Bayesian optimal algorithm. The Lemma would hold for any algorithm which does not stop on seeing a bad day of weather. With this, we can now prove \Cref{l:utilitycloseness} which says that for probabilities that are close to each other, for any algorithm $\pialg$, the utilities $\frac{\pialg(\skidist)}{\skibayopt(\skidist)}$ are close to each other.

\begin{proof}[Proof of \Cref{l:utilitycloseness}]
First consider when $\skidist_{\alpha},\skidist_{\beta}$ lie in the interval $\left[\delta,\frac{\stopcost-1}{\inplen-1}\right]$. The Bayesian optimal cost in this interval is $\skibayopt(\skidist) = T\skidist$. By assumption, the algorithm does not stop on seeing a bad day of weather and thus $\pialg(\{0\}^{\inplen})=0$. We can then write the utility as

\begin{equation}
\label{e:utilform}
    \begin{split}
        \util(\pialg,\skidist) &= \frac{1}{\inplen}\sum_{\piinput \in \{0,1\}^{\inplen}\backslash \{0\}^{\inplen}}\pialg(\piinput)\skidist^{\sum_{\timeindex=1}^{\inplen}\inpelement{\timeindex}-1}(1-\skidist)^{\inplen-\sum_{\timeindex=1}^{\inplen}\inpelement{\timeindex}} 
    \end{split}
\end{equation}

Observe that in its current form, the utility is not an expectation over a feasible probability distribution and thus the arguments of \autoref{l:algoperfcloseness} do not pass through as is. We re-write the utility function in a manner such that it can be written as the expectation of a function over a feasible distribution. Our goal is to write the utility in a form such that it is generated by a binomial random variable of length $(\inplen-1)$.

\begin{equation}
\label{e:expecform}
\begin{split}
    \mathbf{E}_{\piinputmod \sim \Delta(\{0,1\}^{\inplen-1})}[\pimodalg(\piinputmod)] & = \sum_{\piinputmod \in \{0,1\}^{\inplen-1}}\pimodalg(\piinputmod)\skidist^{\sum_{\timeindex=1}^{\inplen-1}\inpelement{\timeindex}}(1-\skidist)^{\inplen-1-\sum_{\timeindex=1}^{\inplen-1}\inpelement{\timeindex}}
\end{split}
\end{equation}

The idea is to club terms together in \autoref{e:utilform} so as to represent the form above. For example, coefficients for sequences $\{\piinput : \sum_{i=1}^{\inplen}\inpelement{\timeindex}=1\}$ are clubbed together to give 

$$\pimodalg(\{0\}^{\inplen-1}) = \sum_{\{\piinput : \sum_{\timeindex=1}^{\inplen}\inpelement{\timeindex}=1\}}\pialg(\piinput)$$

In general, for a $r \in [\inplen-1]$, the number of sequences such that $\sum_{\timeindex=1}^{\inplen-1}\inpelementmod{\timeindex}=r$ are $\binom{\inplen-1}{r}$ in \autoref{e:expecform}, while the number of sequences such that $\sum_{\timeindex=1}^{\inplen}\inpelement{\timeindex}=r+1$ are $\binom{\inplen}{r+1}$ in \autoref{e:utilform}. Thus for each term in \autoref{e:expecform} we need to club $\frac{\binom{\inplen}{r+1}}{\binom{\inplen-1}{r}} = \frac{\inplen}{r+1}$ terms in \autoref{e:utilform}. A process to do this would be, for $r \in \{0\} \cup [\inplen-1]$, for the set sequences $S_{r}=\{(\inpelementmod{1},\dots,\inpelementmod{\inplen-1}) : \sum_{\timeindex=1}^{\inplen-1}\inpelementmod{\timeindex}=r\}$, consider the set of sequences $S_{r}^{*} = \{(\inpelement{1},\dots,\inpelement{\inplen}) : \sum_{\timeindex=1}^{\inplen}\inpelement{\timeindex}=r+1\}$. Consider any ordering of the sequences in $S_{r}$ and assign the elements indices in $\left[\binom{\inplen-1}{r}\right]$. Consider any ordering of the sequences in $S_{r}^{*}$ and assign the elements indices in $\left[\binom{\inplen}{r+1}\right]$. Now for sequence indexed $k$ in $S_{r}$, which we specifically label $\seqkrmod$, assign the set of sequences indexed $S_{r}^{*^{k}} = \left\{(k-1)*\ceil{\frac{\inplen}{r+1}}+1,\dots,\max\{k*\ceil{\frac{\inplen}{r+1}},\binom{\inplen}{r+1}\}\right\}$ in $S_{*}^{r}$. If $k$ is such that $(k-1)*\ceil{\frac{\inplen}{r+1}}+1 > \binom{\inplen}{r+1}$, then $S_{r}^{*^{k}}$ is the null set. Now we define the function $\pimodalg(.)$ as

$$\pimodalg(\seqkrmod) = \sum_{\piinput \in S_{r}^{*^{k}}}\pialg(\piinput)$$

Importantly, note that $\pimodalg(\seqkrmod) \leq \ceil{\frac{\inplen}{r+1}}\max \pialg(\piinput) \leq \ceil{\frac{\inplen}{r+1}}(\stopcost+r+1)$, where the second inequality holds because the assignment is over paths that observe $(r+1)$ good days and thus the maximum cost suffered by any algorithm on this path is at most $(\stopcost+r+1)$. With this assignment, we have that

$$\util(\pialg,\skidist) = \frac{1}{T}\mathbf{E}_{\piinputmod \sim \Delta(\{0,1\}^{\inplen-1})}[\pimodalg(\piinputmod)]$$

where the distribution is generated by a binomial process with probability $\skidist$ and length $(\inplen-1)$. Applying \autoref{l:algoperfcloseness}, we have that if $|\skidist_{\alpha}-\skidist_{\beta}| \leq \frac{\epsilon}{2(\stopcost+1)^{2}(\inplen+1)^{2}\inplen}$, when $\skidist_{\alpha},\skidist_{\beta} \in \left[\delta,\frac{\stopcost-1}{\inplen-1}\right]$,

$$|\util(\pialg,\skidist_{\alpha})-\util(\pialg,\skidist_{\beta})| \leq \frac{\epsilon}{2\inplen(\stopcost+1)(\inplen+1)} \leq \epsilon$$

where the second inequality follows from that when we apply \autoref{l:algoperfcloseness}, the worst case cost along a path is upper bounded as $\max_{r}\ceil{\frac{\inplen}{r+1}}(\stopcost+r+1)$ which is maximized at $r=0$. Now consider the case $\skidist_{\alpha},\skidist_{\beta}$ in the interval $\left(\frac{\stopcost-1}{\inplen-k},\min\left\{\frac{\stopcost-1}{\inplen-k-1},1\right\}\right]$ such that
    
    $$|\skidist_{\alpha}-\skidist_{\beta}| \leq \frac{\epsilon}{2(\stopcost+1)^{2}(\inplen+1)^{2}\inplen}$$

We have that

\begin{equation*}
    \begin{split}
        |\util(\pialg,\skidist_{\alpha})-\util(\pialg,\skidist_{\beta})| & = \left |\frac{\pialg(\skidist_{\alpha})}{\skibayoptspec{\alpha}(\skidist_{\alpha})} - \frac{A(\skidist_{\beta})}{\skibayoptspec{\beta}(\skidist_{\beta})}\right| \\
        & = \left |\frac{\pimodalg(\skidist_{\alpha})}{\decompfunc(\skidist_{\alpha})} - \frac{\pimodalg(\skidist_{\beta})}{\decompfunc(\skidist_{\beta})}\right| \\
        & \leq \epsilon
    \end{split}
\end{equation*}

where the second equality follows from \autoref{l:appendixbayoptdecompose}. Note that $\decompfunc$ can also be decomposed in an expected form and thus the Lipschitz bound derived above holds for $\decompfunc(\skidist)$ as well. Thus the inequality follows from noting that $\decompfunc(.) \geq 1$ for all $\skidist \in [\delta,1]$ (\autoref{l:appendixbayoptdecompose}), and an upper bound on both $\pimodalg(.)$ and $\decompfunc(.)$ is $\max_{r}\ceil{\frac{\inplen}{r+1}}(\stopcost+r+1)$ which is upper bounded by $(\stopcost+1)(\inplen+1)$. 
\end{proof}

\begin{proof}[Proof of \autoref{t:epscoverskirental}]
The idea is to discretize each interval where the Bayesian optimal algorithm is the same (\autoref{t:bayopt}), with a discretization of size
$$\Bar{\epsilon} = \frac{\epsilon}{8(\stopcost+1)^{2}(\inplen+1)^{2}\inplen}.$$ First consider the interval $\left[\delta,\frac{\stopcost-1}{\inplen-1}\right]$, where the Bayesian optimal algorithm always rents. Consider the following grid points on this interval
$$\skigrid{0} = \left\{\delta + k\Bar{\epsilon}\right\}_{k=1}^{\ceil{\frac{1}{\Bar{\epsilon}}\left[\frac{\stopcost-1}{\inplen-1}-\delta\right]-1}}.$$

\noindent
Such a discretization covers the interval $\left[\delta,\frac{\stopcost-1}{\inplen-1}\right]$. We have that for any grid point in $\skigrid{0}$ i.e. $(\skidisteps)_{k} \in \skigrid{0}$,

$$|\util(\pialg,(\skidisteps)_{k}) - \util(\pialg,\skidist)| \leq \frac{\epsilon}{4} \quad \forall p \in \left[(\skidisteps)_{k}-\Bar{\epsilon},(\skidisteps)_{k}+\Bar{\epsilon}\right]$$

\noindent
which follows from \Cref{l:utilitycloseness}. We do this process for every interval $\left[\frac{\stopcost-1}{\inplen-l},\min\left\{\frac{\stopcost-1}{\inplen-l-1},1\right\}\right]$, where $l \in [\min\{\ceil{\inplen-\stopcost},\inplen-1\}]$ to obtain a grid for that interval $\skigrid{l}$.

$$\skigrid{l} = \left\{\frac{\stopcost-1}{\inplen - l} + k\Bar{\epsilon}\right\}_{k=1}^{\ceil{\frac{1}{\Bar{\epsilon}}\left[\min\left\{\frac{\stopcost-1}{\inplen - l -1},1\right\}-\frac{\stopcost-1}{\inplen - l}\right]-1}}$$.

\noindent
By the same argument for $\skigrid{0}$, these satisfy the same properties as well. Hence 

$$\skidistclassepsgrid = \cup_{l=0}^{\min\{\ceil{\inplen-\stopcost},\inplen-1\}}\skigrid{l}$$ 

\noindent
satisfies Small-Cover.
\end{proof}

\begin{proof}[Proof of \autoref{t:skirentalfptas}]
To compute the run-time, first note that the number of actions to the distribution player is bounded as $|\skidistclassepsgrid| \leq \frac{2(\stopcost+1)^{2}(\inplen+1)^{2}\inplen}{\epsilon}$ (look at proof of \autoref{t:epscoverskirental}). Also note that the maximum possible utility encountered by the learning algorithm is $(B+1)(T+1)$ (look at the proof of \Cref{l:utilitycloseness} in \autoref{a:proofs-framework}). The number of updates the algorithm takes is $\frac{(\stopcost+1)^{2}(\inplen+1)^{2} \log |\skidistclassepsgrid|}{\epsilon^{2}}$. Each update takes $O(|\skidistclassepsgrid|\inplen^{3})$ time to execute (\autoref{t:polybackwardsinduction}). Thus the overall running time is $\Tilde{O}(\inplen^{12}\epsilon^{-3})$ where we have used the inequality $\stopcost \leq \inplen$ as the problem is trivial when $\stopcost > \inplen$ in which case the agent always continues. 
\end{proof}
\section{Proofs For Section 5}
\label{s:empiricalproofs}

\begin{proof}[Proof of Lemma \ref{l:commonworstcaseinput}]
By definition, the sub-game optimal algorithm takes an optimal decision at every stage of the game. Since the worst-case optimal algorithm takes an optimal decision on some input sequence, it follows that the sets have a non-empty intersection.
\end{proof}

\section{Properties of the Optimal Prior-Independent Ski-Rental Algorithm}
\label{a:optpiproperties}

In this section, we describe certain properties of an optimal prior-independent ski-rental algorithm. We first define what we mean by an algorithm to satisfy the thresholding property. Intuitively, for any sub-sequence of days of a fixed length, if the algorithm stops on seeing a certain count of good days, then it will also stop on seeing a greater count of good days. We then show a property of the worst-case distribution that the adversary plays.

\begin{definition}
(Thresholding Algorithm) We say that an algorithm $\pialg=(\alg{1},\dots,\alg{\inplen})$ is a thresholding algorithm if $\forall \timeindex \in [\inplen]$, $ \forall \inpelement{\timeindex} \in \{0,1\}$, if $ \exists k \in \{0\} \cup [\timeindex-1]$ such that $\alg{i}(k,\inpelement{\timeindex}) = \stp$, then $\alg{i}(k+1,\inpelement{\timeindex}) = \stp$
\end{definition}

\begin{lemma}
\label{l:thresholdingbestresponse}
For any distribution, $\skidistoverdist \in \Delta([\delta,1])$ that the adversary plays, an optimal algorithm for that distribution satisfies the thresholding property.  
\end{lemma}

\begin{proof}
We show the above Lemma for any discrete distribution, but the same arguments go through for a continuous distribution as well. We use the same notation as that used in \autoref{t:polybackwardsinduction}. We do an induction over the length of sub-sequences $n \in [T]$, while maintaining the following invariant at each stage

\begin{itemize}
    \item The function $u_{t}(k)$, where $u_{t}(k)$ is given as,
    $$u_{t}(k) = \frac{\sum_{p \in supp(P)}\frac{Pr[(k,t,1)|p]Pr[\skidistoverdist=\skidist]}{OPT(p)}c^{p}(k,t,1)}{\sum_{p \in supp(P)}\frac{Pr[(k,t,1)|p]Pr[\skidistoverdist=\skidist]}{OPT(p)}}$$
    is monotone non-decreasing in $k$.
    \item The function $l_{t}(k)$, where $l_{t}(k)$ is given as,
     $$l_{t}(k) = \frac{\sum_{p \in supp(P)}\frac{Pr[(k,t,0)|p]Pr[\skidistoverdist=\skidist]}{OPT(p)}c^{p}(k,t,0)}{\sum_{p \in supp(P)}\frac{Pr[(k,t,0)|p]Pr[\skidistoverdist=\skidist]}{OPT(p)}}$$
    is monotone non-decreasing in $k$.
    \item $u_{t}(k) \geq l_{t}(k)$ for all $k \in \{0\} \cup [t-1]$. 
\end{itemize}

Let us first check the base case when $t=T$. We always continue on the last day no matter what and thus we have that $c^{p}(k,T,1) = 1$ and $c^{p}(k,T,0) = 0$ for all $k \in \{0\} \cup [T-1]$ and $p \in supp(P)$. Thus the base case is satisfied. Also observe that thresholding is vacuously satisfied at stage $t=T$. We now prove that when the induction assumptions are satisfied at stage $(n+1)$, then we get the thresholding property of the algorithm at stage $n$. It is sufficient to prove that on observing a good day, the normalized continuation cost of continueing is monotone non-decreasing in $k$. This is because the algorithm takes a decision according to 

\begin{multline*}
    A^{P}_{t}(k,1) = arg\,min_{\{B,R\}} \{\sum_{p \in supp(P)}Pr[(k,t,1)|p]Pr[\skidistoverdist=\skidist]\frac{B}{OPT(p)}, \\ \sum_{p \in supp(P)}Pr[(k,t,1)|p]Pr[\skidistoverdist=\skidist]\left[\frac{1 + pc^{p}(k+1,t+1,1) + (1-p)c^{p}(k+1,t+1,0)}{OPT(p)}\right]  \}
\end{multline*}

or effectively compares the stop cost $B$ against 

$$u^{R}_{t}(k) = \frac{\sum_{p \in supp(P)}\frac{Pr[(k,t,1)|p]Pr[\skidistoverdist=\skidist]}{OPT(p)}\left[1 + pc^{p}(k+1,t+1,1) + (1-p)c^{p}(k+1,t+1,0)\right]}{\sum_{p \in supp(P)}\frac{Pr[(k,t,1)|p]Pr[\skidistoverdist=\skidist]}{OPT(p)}}$$

Thus if the latter is monotone non-decreasing, if at any $k$ it becomes greater than $B$, it will still be greater than $B$ for larger $k$. We ignore writing the limits of the summation as it is clear it is over probabilites in the support. To make the calculations more compact, we re-write the weighted probability as

$$f^{p}(k,t,\inpelement{\inplen}) = \frac{Pr[(k,t,\inpelement{\inplen})|p]Pr[\skidistoverdist=\skidist]}{OPT(p)}, \quad \inpelement{\inplen} \in \{0,1\}$$

Our aim is to prove $u^{R}_{t}(k+1) - u^{R}_{t}(k) \geq 0$. We have

\begin{multline*}
     u^{R}_{t}(k+1) - u^{R}_{t}(k) = \overbrace{\sum p\left[ \frac{f^{p}(k+1,t,1)}{\sum f^{p}(k+1,t,1)}c^{p}(k+2,t+1,1)  - \frac{f^{p}(k,t,1)}{\sum f^{p}(k,t,1)} c^{p}(k+1,t+1,1) \right]}^{\encircled{1}} + \\ \underbrace{\sum (1-p) \left[ \frac{f^{p}(k+1,t,1)}{\sum f^{p}(k+1,t,1)}c^{p}(k+2,t+1,0)  - \frac{f^{p}(k,t,1)}{\sum f^{p}(k,t,1)} c^{p}(k+1,t+1,0) \right]}_{\encircled{2}}
\end{multline*}

Considering the first summation, it simplifies to the following equation which follows from the i.i.d assumption on the class of distributions we consider

\begin{multline}
\label{e:sumponesimple}
   \encircled{1} = \frac{\sum f^{p}(k+2,t+1,1)}{\sum f^{p}(k+1,t,1)}\left [ \sum \frac{f^{p}(k+2,t+1,1) c^{p}(k+2,t+1,1)}{\sum f^{p}(k+2,t+1,1)} \right ] - \\  \frac{\sum f^{p}(k+1,t+1,1)}{\sum f^{p}(k,t,1)} \left[ \sum \frac{f^{p}(k+1,t+1,1) c^{p}(k+1,t+1,1)}{\sum f^{p}(k+1,t+1,1)} \right ]
\end{multline}

Utilizing the inductive assumption corresponding to the non-decreasing property of $u_{t+1}(k)$, we have that 

\begin{equation*}
    \begin{split}
        \encircled{1}  & \geq \left[\sum \frac{f^{p}(k+1,t+1,1) c^{p}(k+1,t+1,1)}{\sum f^{p}(k+1,t+1,1)} \right] \overbrace{\left[ \frac{\sum f^{p}(k+2,t+1,1)}{\sum f^{p}(k+1,t,1)} - \frac{\sum f^{p}(k+1,t+1,1)}{\sum f^{p}(k,t,1)} \right]}^{\encircled{1b}}
    \end{split}
\end{equation*}

Analyzing the term $\encircled{1b}$, we have

\begin{equation*}
\label{e:lbpartoneb}
    \begin{split}
        \encircled{1b} &= \frac{\left(\sum f^{p}(k+2,t+1,1)\right)\left(\sum f^{p}(k,t,1)\right) - \left(\sum f^{p}(k+1,t+1,1)\right)\left(\sum f^{p}(k+1,t,1)\right)}{\left(\sum f^{p}(k+1,t,1)\right)\left(\sum f^{p}(k,t,1)\right)} \\
        & = \frac{\sum_{\{(p_{1},p_{2}) : p_{1} < p_{2}\}}\frac{Pr[P=p_{1}]Pr[P=p_{2}]}{OPT(p_{1})OPT(p_{2})}(p_{1}p_{2})^{k+1}\left[(1-p_{1})(1-p_{2})\right]^{t-k-2}(p_{1}-p_{2})^{2}}{\left(\sum f^{p}(k+1,t,1)\right)\left(\sum f^{p}(k,t,1)\right)}
    \end{split}
\end{equation*}

Analyzing term $\encircled{2}$, we obtain the following simplified equation again following from the i.i.d assumptions on the class of distributions we consider. 

\begin{multline}
\label{e:sumptwosimple}
   \encircled{2} = \frac{\sum f^{p}(k+2,t+1,0)}{\sum f^{p}(k+1,t,1)}\left [ \sum \frac{f^{p}(k+2,t+1,0) c^{p}(k+2,t+1,0)}{\sum f^{p}(k+2,t+1,0)} \right ] - \\  \frac{\sum f^{p}(k+1,t+1,0)}{\sum f^{p}(k,t,1)} \left[ \sum \frac{f^{p}(k+1,t+1,0) c^{p}(k+1,t+1,0)}{\sum f^{p}(k+1,t+1,0)} \right ]
\end{multline}

Utilizing the inductive assumption corresponding to the non-decreasing property of $l_{t+1}(k)$, we have that 

\begin{equation*}
    \begin{split}
        \encircled{2}  & \geq \left[\sum \frac{f^{p}(k+1,t+1,0) c^{p}(k+1,t+1,0)}{\sum f^{p}(k+1,t+1,0)} \right] \overbrace{\left[ \frac{\sum f^{p}(k+2,t+1,0)}{\sum f^{p}(k+1,t,1)} - \frac{\sum f^{p}(k+1,t+1,0)}{\sum f^{p}(k,t,1)} \right]}^{\encircled{2b}}
    \end{split}
\end{equation*}

Analyzing the term $\encircled{2b}$, we have 

\begin{equation*}
\label{e:lbparttwob}
    \begin{split}
        \encircled{2b} &= \frac{\left(\sum f^{p}(k+2,t+1,0)\right)\left(\sum f^{p}(k,t,1)\right) - \left(\sum f^{p}(k+1,t+1,0)\right)\left(\sum f^{p}(k+1,t,1)\right)}{\left(\sum f^{p}(k+1,t,1)\right)\left(\sum f^{p}(k,t,1)\right)} \\
        & = -\frac{\sum_{\{(p_{1},p_{2}):p_{1} < p_{2}\}}\frac{Pr[P=p_{1}]Pr[P=p_{2}]}{OPT(p_{1})OPT(p_{2})}(p_{1}p_{2})^{k+1}\left[(1-p_{1})(1-p_{2})\right]^{t-k-2}(p_{1}-p_{2})^{2}}{\left(\sum f^{p}(k+1,t,1)\right)\left(\sum f^{p}(k,t,1)\right)} \\
        &= -\encircled{1b}
    \end{split}
\end{equation*}

Thus combining $\encircled{1}$ and $\encircled{2}$ to give a bound on $u^{R}_{t}(k+1) - u^{R}_{t}(k)$, we have

\begin{equation*}
    \begin{split}
        u^{R}_{t}(k+1) - u^{R}_{t}(k) & \geq \encircled{1b}\left[\sum \frac{f^{p}(k+1,t+1,1) c^{p}(k+1,t+1,1)}{\sum f^{p}(k+1,t+1,1)} - \sum \frac{f^{p}(k+1,t+1,0) c^{p}(k+1,t+1,0)}{\sum f^{p}(k+1,t+1,0)}  \right] \\
        & \geq 0
    \end{split}
\end{equation*}

which follows from $\encircled{1b}$ being non-negative and the inductive assumption that $u_{t+1}(k) \geq l_{t+1}(k)$. 

All that is left to verify is that the inductive invariants are maintained at stage $t$. Since the algorithm is thresholding at stage $t$, let the algorithm continue till a count of $k^{*}$ and then stop for strictly greater than $k^{*}$ good days. If $k^{*}=(t-1)$, then the inductive invariant holds as $u_{t}(k) = u^{R}_{t}(k)$ for all $k \in \{0\} \cup [t-1]$, which we have proved to be monotonic non-decreasing in $k$. If $k^{*} < t$, then $u_{t}(k)$ is monotone non-decreasing for $k \in \{0\} \cup [k^{*}]$. We also have that $u_{t}(k)$ is monotone non-decreasing for $k \in \{k^{*}+1,t-1\}$ as $c^{p}(k,t,1) = B$ for all $p \in supp(P)$. All that is remaining to show is that $u_{t}(k^{*}+1) \geq u_{t}(k^{*})$.

$$u_{t}(k^{*}+1) =  \frac{\sum_{p \in supp(P)}\frac{Pr[(k^{*}+1,t,1)|p]Pr[\skidistoverdist=\skidist]}{OPT(p)}B}{\sum_{p \in supp(P)}\frac{Pr[(k^{*}+1,t,1)|p]Pr[\skidistoverdist=\skidist]}{OPT(p)}} = B > u_{t}^{R}(k^{*}) = u_{t}(k^{*})$$

where the inequality follows from the fact that the algorithm chose to continue on seeing $k^{*}$ good days in the first $(t-1)$ days. Now to verify the monotone non-decreasing nature of $l_{t}(k)$, we follows the same procedure as that done to show the monotonicity of $u_{t}^{R}(k)$. Note that the algorithm always continues on seeing a bad day of weather. Thus,

\begin{multline*}
     l_{t}(k+1) - l_{t}(k) = \overbrace{\sum p\left[ \frac{f^{p}(k+1,t,0)}{\sum f^{p}(k+1,t,0)}c^{p}(k+1,t+1,1)  - \frac{f^{p}(k,t,0)}{\sum f^{p}(k,t,0)} c^{p}(k,t+1,1) \right]}^{\encircled{3}} + \\ \underbrace{\sum (1-p) \left[ \frac{f^{p}(k+1,t,0)}{\sum f^{p}(k+1,t,0)}c^{p}(k+1,t+1,0)  - \frac{f^{p}(k,t,0)}{\sum f^{p}(k,t,0)} c^{p}(k,t+1,0) \right]}_{\encircled{4}}
\end{multline*}

Considering the first summation, it simplifies to the following equation which follows from the i.i.d assumption on the class of distributions we consider

\begin{multline}
\label{e:sumpthreesimple}
   \encircled{3} = \frac{\sum f^{p}(k+1,t+1,1)}{\sum f^{p}(k+1,t,0)}\left [ \sum \frac{f^{p}(k+1,t+1,1) c^{p}(k+1,t+1,1)}{\sum f^{p}(k+1,t+1,1)} \right ] - \\  \frac{\sum f^{p}(k,t+1,1)}{\sum f^{p}(k,t,0)} \left[ \sum \frac{f^{p}(k,t+1,1) c^{p}(k,t+1,1)}{\sum f^{p}(k,t+1,1)} \right ]
\end{multline}

Utilizing the inductive assumption at stage $(t+1)$ on the monotonicity of $u_{t+1}(k)$, we have that

\begin{equation*}
    \begin{split}
        \encircled{3}  & \geq \left[\sum \frac{f^{p}(k,t+1,1) c^{p}(k,t+1,1)}{\sum f^{p}(k,t+1,1)} \right] \overbrace{\left[ \frac{\sum f^{p}(k+1,t+1,1)}{\sum f^{p}(k+1,t,0)} - \frac{\sum f^{p}(k,t+1,1)}{\sum f^{p}(k,t,0)} \right]}^{\encircled{3b}}
    \end{split}
\end{equation*}

Analyzing the term $\encircled{3b}$, we have

\begin{equation*}
\label{e:lbpartthreeb}
    \begin{split}
        \encircled{3b} &= \frac{\left(\sum f^{p}(k+1,t+1,1)\right)\left(\sum f^{p}(k,t,0)\right) - \left(\sum f^{p}(k,t+1,1)\right)\left(\sum f^{p}(k+1,t,0)\right)}{\left(\sum f^{p}(k+1,t,0)\right)\left(\sum f^{p}(k,t,0)\right)} \\
        & \geq \frac{\sum_{\{(p_{1},p_{2}) : p_{1} < p_{2}, p_{2} \neq 1\}}\frac{Pr[P=p_{1}]Pr[P=p_{2}]}{OPT(p_{1})OPT(p_{2})}(p_{1}p_{2})^{k}\left[(1-p_{1})(1-p_{2})\right]^{t-k-1}(p_{1}-p_{2})^{2}}{\left(\sum f^{p}(k+1,t,0)\right)\left(\sum f^{p}(k,t,0)\right)}
    \end{split}
\end{equation*}

Analyzing term $\encircled{4}$, we obtain the following simplified equation again following from the i.i.d assumptions on the class of distributions we consider. 

\begin{multline}
\label{e:sumpfoursimple}
   \encircled{4} = \frac{\sum f^{p}(k+1,t+1,0)}{\sum f^{p}(k+1,t,0)}\left [ \sum \frac{f^{p}(k+1,t+1,0) c^{p}(k+1,t+1,0)}{\sum f^{p}(k+1,t+1,0)} \right ] - \\  \frac{\sum f^{p}(k,t+1,0)}{\sum f^{p}(k,t,0)} \left[ \sum \frac{f^{p}(k,t+1,0) c^{p}(k,t+1,0)}{\sum f^{p}(k,t+1,0)} \right ]
\end{multline}

Utilizing the inductive assumption corresponding to the non-decreasing property of $l_{t+1}(k)$, we have that 

\begin{equation*}
    \begin{split}
        \encircled{4}  & \geq \left[\sum \frac{f^{p}(k,t+1,0) c^{p}(k,t+1,0)}{\sum f^{p}(k,t+1,0)} \right] \overbrace{\left[ \frac{\sum f^{p}(k+1,t+1,0)}{\sum f^{p}(k+1,t,0)} - \frac{\sum f^{p}(k,t+1,0)}{\sum f^{p}(k,t,0)} \right]}^{\encircled{4b}}
    \end{split}
\end{equation*}

Analyzing the term $\encircled{4b}$, we have 

\begin{equation*}
\label{e:lbpartfourb}
    \begin{split}
        \encircled{4b} &= \frac{\left(\sum f^{p}(k+1,t+1,0)\right)\left(\sum f^{p}(k,t,0)\right) - \left(\sum f^{p}(k,t+1,0)\right)\left(\sum f^{p}(k+1,t,0)\right)}{\left(\sum f^{p}(k+1,t,0)\right)\left(\sum f^{p}(k,t,0)\right)} \\
        & = -\frac{\sum_{\{(p_{1},p_{2}):p_{1} < p_{2}, p_{2} \neq 1\}}\frac{Pr[P=p_{1}]Pr[P=p_{2}]}{OPT(p_{1})OPT(p_{2})}(p_{1}p_{2})^{k}\left[(1-p_{1})(1-p_{2})\right]^{t-k-1}(p_{1}-p_{2})^{2}}{\left(\sum f^{p}(k+1,t,1)\right)\left(\sum f^{p}(k,t,1)\right)} \\
        & \geq -\encircled{3b}
    \end{split}
\end{equation*}

Thus combining $\encircled{3}$ and $\encircled{4}$ to give a bound on $l_{t}(k+1) - l_{t}(k)$, we have

\begin{equation*}
    \begin{split}
        l_{t}(k+1) - l_{t}(k) & \geq \encircled{3b}\left[\sum \frac{f^{p}(k,t+1,1) c^{p}(k,t+1,1)}{\sum f^{p}(k,t+1,1)} - \sum \frac{f^{p}(k,t+1,0) c^{p}(k,t+1,0)}{\sum f^{p}(k,t+1,0)}  \right] \\
        & \geq 0
    \end{split}
\end{equation*}

which follows from $\encircled{3b}$ being non-negative and the inductive assumption that $u_{t+1}(k) \geq l_{t+1}(k)$. 

All that is left to prove that the third induction assumption is maintained i.e $u_{t}(k) \geq l_{t}(k) \forall k \in \{0\} \cup [t-1]$. First note that if at some $k^{*} \in \{0\} \cup [t-1]$, the algorithm stops, then we clearly have for all $k \geq k^{*}$, 

$$u_{t}(k) = \frac{\sum_{p \in supp(P)}\frac{Pr[(k,t,1)|p]Pr[\skidistoverdist=\skidist]}{OPT(p)}B}{\sum_{p \in supp(P)}\frac{Pr[(k,t,1)|p]Pr[\skidistoverdist=\skidist]}{OPT(p)}} = B > l_{t}(k)$$

where the first equality follows from the fact that $c^{p}(k,t,1)=B \quad \forall p \in supp(P), \forall k \geq k^{*}$ and the inequality from the fact that the agent never stops on seeing a bad day in the Backwards Induction algorithm and thus it must have been that 

\begin{equation*}
    \begin{split}
        B\sum_{p \in supp(P)}\frac{Pr[(k,t,0)|p]Pr[\skidistoverdist=\skidist]}{OPT(p)} &> \sum_{p \in supp(P)}\frac{Pr[(k,t,0)|p]Pr[\skidistoverdist=\skidist]}{OPT(p)}[pc^{p}(k,t+1,1)+(1-p)c^{p}(k,t+1,0)] \\
        & = \sum_{p \in supp(P)}\frac{Pr[(k,t,0)|p]Pr[\skidistoverdist=\skidist]}{OPT(p)}c^{p}(k,t,0)
    \end{split}
\end{equation*}

For $k < k^{*}$, i.e those count of days where the agent continues, we have that 

\begin{multline*}
     u_{t}(k) - l_{t}(k) = 1 + \overbrace{\sum p\left[ \frac{f^{p}(k,t,1)}{\sum f^{p}(k,t,1)}c^{p}(k+1,t+1,1)  - \frac{f^{p}(k,t,0)}{\sum f^{p}(k,t,0)} c^{p}(k,t+1,1) \right]}^{\encircled{5}} + \\ \underbrace{\sum (1-p) \left[ \frac{f^{p}(k,t,1)}{\sum f^{p}(k,t,1)}c^{p}(k+1,t+1,0)  - \frac{f^{p}(k,t,0)}{\sum f^{p}(k,t,0)} c^{p}(k,t+1,0) \right]}_{\encircled{6}}
\end{multline*}

Considering \encircled{5}, using the i.i.d assumption on the class of distributions we consider, 

\begin{multline}
\label{e:sumpfivesimple}
   \encircled{5} = \frac{\sum f^{p}(k+1,t+1,1)}{\sum f^{p}(k,t,1)}\left [ \sum \frac{f^{p}(k+1,t+1,1) c^{p}(k+1,t+1,1)}{\sum f^{p}(k+1,t+1,1)} \right ] - \\  \frac{\sum f^{p}(k,t+1,1)}{\sum f^{p}(k,t,0)} \left[ \sum \frac{f^{p}(k,t+1,1) c^{p}(k,t+1,1)}{\sum f^{p}(k,t+1,1)} \right ]
\end{multline}

Utilizing the inductive assumption at stage $(t+1)$ on the monotonicity of $u_{t+1}(k)$, we have that

\begin{equation*}
    \begin{split}
        \encircled{5}  & \geq \left[\sum \frac{f^{p}(k,t+1,1) c^{p}(k,t+1,1)}{\sum f^{p}(k,t+1,1)} \right] \overbrace{\left[ \frac{\sum f^{p}(k+1,t+1,1)}{\sum f^{p}(k,t,1)} - \frac{\sum f^{p}(k,t+1,1)}{\sum f^{p}(k,t,0)} \right]}^{\encircled{5b}}
    \end{split}
\end{equation*}

If the prior distribution had no mass on $p=1$, then observe that $\encircled{5b} = \encircled{3b}$. In the case the adversary put some mass on $p=1$, 

$$\twopartdef{$\encircled{5b} = \encircled{3b}$}{k < t-1}{\encircled{5b} = \encircled{3b} + \frac{\frac{Pr[P=1]}{B}\sum_{p \in supp(P) \setminus \{1\}}\frac{Pr[\skidistoverdist=\skidist]}{OPT(p)}p^{k}(1-p)^{t-k+1}}{(\sum f^{p}(k,t,1))(\sum f^{p}(k,t,0))}}{t=k-1}$$

 Analyzing term $\encircled{6}$,

\begin{multline}
\label{e:sumpsixsimple}
   \encircled{6} = \frac{\sum f^{p}(k+1,t+1,0)}{\sum f^{p}(k,t,1)}\left [ \sum \frac{f^{p}(k+1,t+1,0) c^{p}(k+1,t+1,0)}{\sum f^{p}(k+1,t+1,0)} \right ] - \\  \frac{\sum f^{p}(k,t+1,0)}{\sum f^{p}(k,t,0)} \left[ \sum \frac{f^{p}(k,t+1,0) c^{p}(k,t+1,0)}{\sum f^{p}(k,t+1,0)} \right ]
\end{multline}

Utilizing the inductive assumption corresponding to the non-decreasing property of $l_{t+1}(k)$, we have that 

\begin{equation*}
    \begin{split}
        \encircled{6}  & \geq \left[\sum \frac{f^{p}(k,t+1,0) c^{p}(k,t+1,0)}{\sum f^{p}(k,t+1,0)} \right] \overbrace{\left[ \frac{\sum f^{p}(k+1,t+1,0)}{\sum f^{p}(k,t,1)} - \frac{\sum f^{p}(k,t+1,0)}{\sum f^{p}(k,t,0)} \right]}^{\encircled{6b}}
    \end{split}
\end{equation*}

Analyzing the term $\encircled{6b}$, again we observe that if the prior distribution had no mass on $p=1$, then $\encircled{6b} = \encircled{4b}$. In the case the adversary put some mass on $p=1$,

$$\twopartdef{$\encircled{6b} = \encircled{4b}$}{k < t-1}{\encircled{6b} = \encircled{4b} - \frac{\frac{Pr[P=1]}{B}\sum_{p \in supp(P) \setminus \{1\}}\frac{Pr[\skidistoverdist=\skidist]}{OPT(p)}p^{k}(1-p)^{t-k+1}}{(\sum f^{p}(k,t,1))(\sum f^{p}(k,t,0))}}{t=k-1}$$

Thus we have $\encircled{6b} \geq -\encircled{5b}$. Combining $\encircled{5}$ and $\encircled{6}$, we get that 

\begin{equation*}
    \begin{split}
        u_{t}(k)-l_{t}(k) \geq \encircled{5b}\left[\sum \frac{f^{p}(k,t+1,1) c^{p}(k,t+1,1)}{\sum f^{p}(k,t+1,1)} - \sum \frac{f^{p}(k,t+1,0) c^{p}(k,t+1,0)}{\sum f^{p}(k,t+1,0)} \right] \\
        & \geq 0
    \end{split}
\end{equation*}

\noindent
where the second inequality follows from the fact that $\encircled{5b} \geq 0$ and the inductive assumption $u_{t+1}(k) \geq l_{t+1}(k)$ at stage $(t+1)$. 
\end{proof}

With these individual properties, we can now describe a property of the optimal prior-independent algorithm.

\begin{corollary}
\label{c:infothreshnostopbad}
There exists an optimal prior-independent algorithm such that all deterministic algorithms in its support are information-symmetric, do not stop on seeing bad weather and satisfy the thresholding property.
\end{corollary}

\begin{proof}
Observe that Backwards Induction gave us an algorithm in which two properties were satisfied, \autoref{l:bestresponseinfosymmetric} and \autoref{l:thresholdingbestresponse}. For any mixed strategy that the adversary plays, a best response of the agent is satisfies the above three properties. Applying \autoref{l:reducedgameapxeqlb} then leads to the Corollary. 
\end{proof} 

\autoref{c:infothreshnostopbad} tells us that we can work with a reduced game where strategies of the algorithm satisfy the described three properties. We can then show a property of the worst-case distribution of the adversary in equilibrium. 

\begin{lemma}
\label{l:worstcasedist}
For any given $\inplen$, for a fixed strategy of the agent (an algorithm) which does not stop on seeing bad weather, any best response of the adversary is a discrete distribution.
\end{lemma}

\begin{proof}
\autoref{l:polycost} tells us that for a deterministic algorithm $\pialg$, and a fixed distribution $\skidist \in [\delta,1]$, the cost suffered by the algorithm $\pialg(\skidist)$ is a polynomial in $\skidist$ with degree at most $\inplen$. This would also hold for randomized algorithms as they are nothing but a convex combination over deterministic algorithms. Recall that the expected cost of the algorithm for a distribution $\skidist$, is 

$$\pialg(\skidist) = \sum_{\piinput \in \{0,1\}^{\inplen}}\pialg(\piinput)\skidist^{\sum_{\timeindex=1}^{\inplen}\inpelement{\timeindex}}(1-\skidist)^{\inplen-\sum_{\timeindex=1}^{\inplen}\inpelement{\timeindex}}$$

where $\pialg(\piinput)$ is the cost of the algorithm along a sampled sequence $\piinput=(\inpelement{1},\dots,\inpelement{\inplen})$. 

First note that since we are in a setting where the class of algorithms are those which do not stop on seeing a bad day of weather, the cost of an algorithm cannot be a constant as a function of $\skidist$. This is because $A(\{0,\dots,0\})=0$. The cost of any algorithm is a polynomial of degree at least one. 

Let the algorithms cost be of degree one, that is, $\pialg(\skidist) = \alpha \skidist$ for some $\alpha > 0$. Now for $\skidist \in [\delta,\frac{\stopcost-1}{\inplen-1}]$, the utility of the adversary is $\frac{\alpha}{\inplen}$, while for $\skidist \in (\frac{\stopcost-1}{\inplen-1},\min\left\{\frac{\stopcost-1}{\inplen-2},1\right\}] $, the utility of the adversary is $\frac{\alpha}{\stopcost+(1-\skidist)(\inplen-1)} > \frac{\alpha}{\inplen}$, where the utility was obtained by replacing $\skibayopt(\skidist)$. Observe that the utility is strictly increasing in $\skidist$ and thus is maximized at $\min\left\{\frac{\stopcost-1}{\inplen-2},1\right\}$. For each interval $\skidist \in \left(\frac{\stopcost-1}{\inplen-k},\min\left\{\frac{\stopcost-1}{\inplen-k-1},1\right\}\right]$, where $k \in [\min\{\ceil{\inplen-\stopcost},\inplen-1\}]$, the utility is a ratio of two polynomials in $\skidist$ of degree at most $\inplen$. Thus the number of local extrema of the utility function in each interval are $O(\inplen)$, and thus the number of global maximizers in each interval are $O(\inplen)$. Since a necessary condition of a best response is that the probabilities in the support must be global maximizers in each interval, we have that any best response has at most $O(\inplen^{2})$ probabilities in its support, as there are $O(\inplen)$ intervals. 

Let the algorithms cost be of degree at least two. In that case, the utility in the interval $\skidist \in [\delta,\frac{\stopcost-1}{\inplen-1}]$ is no longer a constant but of degree at least one. The same argument of $O(\inplen)$ local maximizers in each interval follows and we have that any best response has at most $O(\inplen^{2})$ probabilities in its support. 
\end{proof}

\begin{corollary}
\label{c:worstcasedist}
For a given $\inplen$, the worst-case distribution in the prior-independent ski-rental problem is a discrete distribution.
\end{corollary}

\end{document}